\documentclass[reqno]{amsart}
\usepackage{tkz-euclide}

\usepackage[colorlinks,linkcolor=blue]{hyperref}
\usepackage[nameinlink]{cleveref}
\usepackage{amsmath}  
\usepackage{amsthm}
\usepackage{amsfonts}
\usepackage{amssymb} 
\usepackage{mathrsfs} 
\usepackage{graphicx}  
\usepackage{bm}
\usepackage{tikz} 
\usetikzlibrary{svg.path} 
\usepackage{amssymb}
\usepackage{pgfplots}

\usetikzlibrary{calc} 
\usetikzlibrary{arrows.meta} 
\usepackage{amsthm} 
\usepackage{float}
\usepackage{enumitem}
\usepackage[english]{babel}
\usepackage[font=footnotesize, labelfont=bf]{ caption}
\usepackage{ esint }
\usepackage{stackengine,scalerel,graphicx}
\stackMath

\definecolor{trueblue}{rgb}{0.0, 0.45, 0.81}
\definecolor{truegreen}{rgb}{0.13, 0.55, 0.13}

\newcommand{\eps}{\varepsilon}

\theoremstyle{plain}
\begingroup
\newtheorem{theorem}{Theorem}[section]
\newtheorem{lemma}[theorem]{Lemma}

\newtheorem{remark}[theorem]{Remark}

\newtheorem{corollary}[theorem]{Corollary}

\newenvironment{step}[1]{\underline{Step #1}.}{}

\theoremstyle{definition}
\begingroup
\newtheorem{definition}[theorem]{Definition}
\endgroup

\numberwithin{equation}{section}
\newcommand{\N}{\mathbb{N}}
\newcommand{\Z}{\mathbb{Z}}
\newcommand{\R}{\mathbb{R}}

\newcommand{\x}{{\times}}
\newcommand{\defas}{:=}

\usepackage[normalem]{ulem}

\begin{document}

\title[Crystallization in the Winterbottom shape and sharp fluctuation laws]{Crystallization in the Winterbottom shape\\ and sharp fluctuation laws}

%Half-space
%  Crystallization via Stratification\\
%Two-dimensional finite crystallization in presence of a substrate
%(?)\\
%Finite crystallization in the half plane (?) \\
%Finite crystallization in the half plane via stratification (?) \\
% Finite crystallization in presence of a substrate (?) \EEE
%}

\author[M. Friedrich]{Manuel Friedrich} 
\address[Manuel Friedrich]{Department of Mathematics, Friedrich-Alexander Universit\"at Erlangen-N\"urnberg. Cauerstr.~11,
D-91058 Erlangen, Germany}
\email{manuel.friedrich@fau.de}
\urladdr{https://www.math.fau.de/angewandte-mathematik-1/mitarbeiter/prof-dr-manuel-friedrich/}

\author[L.~Kreutz]{Leonard Kreutz}
\address[Leonard Kreutz]{Zentrum Mathematik - M7, Technische Universit\"at M\"unchen,  Boltzmannstr. 3
D-85748 Garching b. M\"unchen,  Germany}
\email{leonard.kreutz@.tum.de}
\urladdr{https://www.math.cit.tum.de/math/personen/wissenschaftliches-personal/kreutz-leonard/}
  
\author[U.~Stefanelli]{Ulisse Stefanelli}
\address[Ulisse Stefanelli]{Faculty of Mathematics, University of
  Vienna, Oskar-Morgenstern-Platz 1, A-1090 Vienna, Austria, \& Vienna
  Research Platform on Accelerating Photoreaction Discovery,
  University of Vienna, Wahringerstraße 17, 1090 Wien, Austria.}
\email{ulisse.stefanelli@univie.ac.at}
\urladdr{https://www.mat.univie.ac.at/~stefanelli/}

%\date{\today}  

\begin{abstract} 
We address finite crystallization in two dimensions in the presence of
a flat crystalline substrate. Particles interact through short-range two-
and three-body potentials favoring local square-lattice
arrangements. An additional interaction term of relative strength $\beta>0$ couples
the particles and the substrate. Our first main result proves
crystallization for all $\beta>0$, corresponding to
 the onset of discrete Winterbottom configurations. 
The
proof relies on a stratification technique from
\cite{FriedrichKreutz}, characterizing the topology of the bond graph
of minimizing configurations.

Our second main result concerns
fluctuations estimates for $\beta\in (0,1)$. We obtain bounds on
the distance between distinct minimizers with the same number $N$ of
particles, showing a sharp scaling law
$N^{3/4}$ when $\beta$ is rational,  and $N^{1/3}$ 
when $\beta$ is irrational and algebraic. This reveals a genuine
substrate-driven effect on fluctuation laws. As a corollary, we derive
a discrete-to-continuum convergence of minimizers towards the
Winterbottom equilibrium shape in the large-particle limit.  
\end{abstract}

\subjclass[2010]{}

\maketitle

\section{Introduction}

At low temperatures, matter typically exhibits crystalline order. In this
regime, the interactions between atoms and molecules can be well
approximated by configurational potentials. In the zero-temperature
limit, it is usually assumed that matter organizes into optimal
configurations, minimizing such configurational energies. Proving
mathematically that such optimal
configurations are periodic ({\it crystalline}) is precisely the aim of the so-called
{\it crystallization problem} ~\cite{Friesecke-Theil15}. Despite considerable attention,
rigorous crystallization results remain scarce. To date,
crystallization of a finite number of particles ({\it finite crystallization}) has been established in one and two spatial dimensions, under various assumptions on the interaction potentials, see Section
\ref{sec:literature} below.

In this paper, we investigate the {\it discrete Winterbottom
  problem}, namely, the determination of equilibrium particle
configurations in a crystal in contact with a
substrate~\cite{Winterbottom}. Specifically, we establish the first crystallization result in the presence
of a flat substrate.  Restricting to two space dimensions, we consider a
finite system of $N$ 
particles interacting at short range through two- and three-body
potentials. These interactions are designed to favor local
arrangements with up to four nearest-neighbors bonds forming $\pi/2$ angles between them. In
addition, the particles interact with a  fixed  flat crystalline substrate.  We are
interested in characterizing optimal configurations of such  
particles. This framework is inspired
by epitaxial growth where the crystal develops layer by layer on a substrate. The relative intensity of the
particle-particle and particle-substrate interaction is described by a
parameter $\beta>0$.

Our results are twofold. At first, we prove crystallization for all
$\beta>0$, see \Cref{thm:crystallization}. To the best of our
knowledge, this constitutes the first two-dimensional
finite-crystallization result {\it relative} to a prescribed
substrate. %\UUU In particular, this delivers the first solvability
%result for the Winterbottom problem in the discrete setting. \EEE 
The crystallization
  proof relies on the {\it stratification} technique introduced in
\cite{FriedrichKreutz}, see Section~\ref{sec:stratification}
below. Due to the substrate
interaction, the minimizing {\it Winterbottom configurations} differ from minimizers 
in the absence of a substrate, in agreement with the predictions of
the Winterbottom problem. Note that for $\beta\geq 1$ the system falls into the so-called
{\it wetting regime}, see \Cref{rem:wetting}, where minimizers flatten against the substrate
and hence are  less interesting. 

Secondly, for $\beta\in(0,1)$ we prove sharp {\it fluctuations}
estimates, see \Cref{thm:fluctuation}. As is typical in
crystallization problems, minimizers are nonunique for most
values of $N$~\cite{DeLuca4,Mainini-Piovano,  Schmidt-disk}.  We prove matching upper and lower bounds on the distance between
minimizers, depending on $\beta$. More precisely,
if $\beta$ is rational, we prove that minimizers can differ up to 
$  C N^{3/4}$ particles for some constant $C>0$. This
scaling coincides with that observed in other two-dimensional lattice regimes without
a substrate~\cite{Davoli15,Davoli16,Mainini-Piovano,Schmidt-disk}. In contrast, for
$\beta$ irrational but algebraic, minimizers may differ up to  $ C_{\beta, \delta} N^{1/3+ \delta }$ particles, for arbitrarily small $ \delta >0$  with $C_{\beta,\delta }\to \infty$ as ${ \delta} \to 0$.   This a
genuine effect of the presence of
the substrate of fluctuation laws. In a much simplified
one-dimensional setting, the sensitivity of ground-state configurations on the value of $\beta$ has been already observed in~\cite{FriedrichStefanelli}.

An immediate
consequence of our  sharp  fluctuation results is the characterization of the
macroscopic shape of minimizers in the large-system limit $N \to
\infty$. We establish a quantitative
discrete-to-continuum convergence of {\it Winterbottom configurations}  towards the
corresponding macroscopic {\it Winterbottom shape}, that is, the
equilibrium shape of a macroscopic crystal in contact with the
substrate.

\subsection{Relation with the literature}\label{sec:literature}
In two dimensions, finite crystallization of hard spheres into the
triangular lattice was first established by  {\sc Heitman \& Radin}~\cite{HR}, building on a result by
 {\sc Harborth}~\cite{Harborth}. The same result was later reobtained
 independently by {\sc De Luca \& Friesecke}~\cite{Lucia}  through
  an entirely different approach, see
 Section~\ref{sec:stratification} below. The case of soft spheres was addressed
by {\sc Radin}~\cite{Radin} and {\sc Wagner}
\cite{Wagner83}, and subsequently extended by {\sc Del Nin \& De Luca}~\cite{LuciaDelNin}. Two-body interactions
in combination with three-body ones can lead either to the square lattice
\cite{Mainini-Piovano} or the hexagonal one~\cite{Mainini}, depending on the
specific choices. Among these, \cite{FriedrichKreutz} introduces a new proof
technique via stratification, which is the one used in this paper as
well,  see
 Section~\ref{sec:stratification} below. Finite crystallization of ionic  dimers  systems has
been obtained both in the square case
 \cite{FriedrichKreutzSquare}  and in the hexagonal case
\cite{FriedrichKreutzHexagonal}. {\sc De Luca, Ninno, \& Ponsiglione}
\cite{DeLuca2} deal with a finite-crystallization
setting including orientations. Apart from the very special setting
considered  in~\cite{Lazzaroni}, no finite-crystallization result in
three  dimensions is currently available. To the best of our
knowledge,  our result in \Cref{thm:crystallization} is the first
finite-crystallization theorem in which the interaction with a substrate is taken into account.

Note that the crystallization problem takes another flavor in the thermodynamic
limit $N \to \infty$. One-dimensional results can be found in
\cite{Gardner,Radi,Ventevogel}, while the stability of periodic one-dimensional configuration, or
lack thereof, is
discussed in~\cite{Hamrick}.
In two dimensions, {\sc Theil}~\cite{Theil} proved that
certain long-range two-body interactions lead to crystallization in the triangular
lattice.  Crystallization as $N \to \infty$ has been obtained by {\sc B\'etermin, De Luca,
  \& Petrache}~\cite{Betermin0} in the square-lattice case and by {\sc
  Farmer, Esedo\=glu, \& Smereka} in the hexagonal one
\cite{Smereka15}.  In three dimensions, {\sc Flatley \& Theil}~\cite{Flateley2}
proved that the face-centered-cubic lattice arises when a specific
selection mechanism on next-to-nearest neighbors interactions is imposed, see also~\cite{Flateley1}. A related computation of the
corresponding {\it Wulff shape} for $N\to \infty$ is
in~\cite{cicalese2}. The defective case has also been
considered. The emergence of rigid polycrystaline structures has been
tackled in \cite{FriedrichKreutzSchmidt}. Dislocations in discrete
structures and their coarse-graining have been studied under different
assumptions on the interaction energy in \cite{Ariza,Alicandro1,Alicandro2,Giuliani},
among others.
Yet another different setting is that of lattice-crystallization,
where one considers  the best lattice for a given lattice
energy. Here, the literature is vast and the reader is referred to
\cite{Betermin2,Betermin3,Cohn, Coulangeon} for references.

Before the present work, the only study dealing with finite
crystallization in the presence of a substrate was the already mentioned~\cite{FriedrichStefanelli}. There, a one-dimensional
hard-sphere finite-crystallization problem is addressed, under the
influence of a periodic background modeling substrate
interaction. In contrast to the setting considered here,
\cite{FriedrichStefanelli} assumes that the substrate and the
crystallizing particles favor different lattice parameters. 

Uniqueness in finite-crystallization problems occurs only for
specific values of $N$. In case of nonuniqueness, the study of the
distance of two distinct minimizers (up to lattice translations) has
led to different results, depending on the underlying lattice
structure. For the triangular lattice,  {\sc
  Schmidt}~\cite{Schmidt-disk} proved  that minimizers can differ
by $CN^{3/4}$ particles, confirming a conjecture from~\cite{Yuen}. This has been revisited in~\cite{Davoli16},
where a sharp constant $C$ is identified.  {\sc Cicalese \&
  Leonardi}~\cite{cicalese} proved the same law by a different approach based
on quantitative isoperimetric inequalities~\cite{Bollobas,Figalli-Maggi-Pratelli,maggi2012sets}. Their method extends to $\Z^d$, as well,
providing the upper bound $CN^{1-1/2d}$. This bound, however, is not
sharp for $d>2$, as was shown by identifying the sharp regime, first in the cubic lattice $\Z^3$~\cite{Mainini-Piovano-schmidt} and eventually  by  {\sc Mainini \& Schmidt}~\cite{edo-bernd}   in any dimension.  The fluctuation law $CN^{3/4}$ has also been shown to hold for the square lattice~\cite{Mainini-Piovano}, the hexagonal lattice~\cite{Davoli15}, and ionic dimer systems~\cite{FriedrichKreutzHexagonal,FriedrichKreutzSquare}.  To the best of our knowledge, \Cref{thm:fluctuation} provides the
first fluctuation estimates in the presence of a substrate.

Under the assumption of crystallinity, the  Winterbottom shape has  already been
identified in the setting of the discrete double-bubble problem in the
square lattice in~\cite{gorny}, see also~\cite{Duncan}, and the
continuous counterparts~\cite{Duncan0,Gorny2,Gorny3}. A
discrete-to-continuous justification of the emergence of a hexagonal Winterbottom shape, emerging from a discrete model with two crystals with mismatched lattice parameters, was given by {\sc Piovano \& Vel\v ci\' c} in~\cite{Piovano1,Piovano2}.
The stability of Winterbottom shapes was investigated in
\cite{Kholmatov, Kotecky,Kreutz}.  In contrast to these works, the novelty of our
contribution lies in proving crystallization rather than assuming
it.   We note, however, that our assumptions on the local geometry of
the crystalline lattice are comparatively more stringent.

\subsection{Stratification technique}\label{sec:stratification} 
In the two-dimensional case, finite-crystallization for hard spheres has often been  obtained
by adapting the {\it induction method over bond-graph layers} by {\sc Heitman \& Radin}
\cite{HR}. This approach determines the fine geometry of
the minimizers by inductively considering the relative
effect of boundary vs.\ bulk particles in the configuration. A second
elegant technique to prove finite-cristallizaton for hard spheres has
been introduced by {\sc De Luca \& Friesecke}~\cite{Lucia},  based on {\it
  discrete geometry}. Here, a notion of discrete combinatorial
curvature is associated to the natural bond graph of the configuration, and a discrete Gau\ss-Bonnet-like
theorem is applied. However, neither of these techniques appears to
extend naturally to the case of particle interactions with a
substrate. This is particularly evident for the induction method over bond-graph layers \cite{HR}, relying on the idea that adding (or removing) a boundary layer to a minimizer should maintain minimality -- a property  which is simply false in the presence of a substrate. 

A key technical difference of the present work compared to earlier
ones is our use of an alternative argument based on   {\it
  stratification}. This method, first introduced by the first two
authors in~\cite{FriedrichKreutz}, provided an alternative
proof  of the finite-crystallization result for the square lattice of~\cite{Mainini-Piovano}.

 We now give a heuristic overview of the stratification technique,
 postponing details to Section~\ref{sec:mainproof} below.  Using purely variational arguments, we first show that minimizers
are {\it regular}: every bond between two neighboring
particles has lenght approximately
$1$, and the angle formed by two adjacent bonds is approximately a
multiple of $\pi/2$. To each such regular configuration we associate
its {\it strata}, namely all bond paths that are approximately
straight. The crystallization result follows by characterizing the
topology of these strata. In particular, due to the presence of the
substrate, we can distinguish between those strata that are {\it interacting} or {\it
  noninteracting} with it. We prove that two interacting strata, or
two noninteracting ones, cannot intersect, and that each interacting stratum crosses all
noninteracting ones. This structural property allows us to reconstruct the global topology of the bond graph of minimizers, ultimately proving \Cref{thm:crystallization}.

\subsection{Discrepancy theory}\label{sec:discrepancy}
Previous results on fluctuation estimates have relied either
on direct manipulations of minimizing configurations
\cite{Davoli15,Mainini-Piovano-schmidt,edo-bernd,Schmidt-disk}
or on quantitative discrete isoperimetric inequalities
\cite{cicalese,Davoli16,FriedrichKreutz, Mainini-Piovano} (or both).  Extending these techniques
to the case of a substrate is not straightforward. On the one hand,
the substrate restricts the possibility of directly manipulating
minimizers. On the other hand, discrete isoperimetric inequalities
in the presence of a substrate are  currently  not available.
 
In this paper, we are hence forced to follow a different path, 
essentially based on {\it discrepancy theory} for sequences
\cite{Kuipers}. This plays a crucial role to treat the
case  $\beta \in (0,1) $ irrational but 
algebraic. In particular, it is used to prove the following fact: for all $N_0$ large enough, one can find at
most  $ C_{\beta,\delta}  N_0^{2/3+ \delta }$ and at least  $  c_{\beta,\delta} N_0^{2/3- \delta }$ 
minimizers with $N$ particles for $N\in [N_0,2N_0]$ taking the form of an
exact {\it rectangle} (i.e., arrangements of particles in $\Z^2$
with equal rows and equal  columns).  Here, $C_\delta \to \infty$ and
$c_\delta \to 0$ as $\delta \to 0$.  Moreover, we prove that any two such
optimal rectangles  differ  in the number
of particles at least by $ c_{\beta,\delta}  N_0^{1/3- \delta }$ and at most by
$  C_{\beta,\delta} N_0^{1/3+ \delta }$. As a consequence, all minimizers with  $N\in
[N_0,2N_0]$ are at most $ C_{\beta,\delta} N_0^{1/3+ \delta }$ far from an optimal
rectangle. At the same time, one finds $N\in [N_0,2N_0]$   such that
the closest optimal rectangle has distance at least
$ c_{\beta,\delta}   N_0^{1/3- \delta }$ and uses such a rectangle to the same lower bound
on the fluctuation.

\subsection{Structure of the paper}
The model is introduced in Section~\ref{sec:setting}, where we also
state the main results. In particular, \Cref{thm:crystallization} and
\Cref{thm:fluctuation} contain the crystallization and the
fluctuation results, respectively. Section~\ref{sec:mainproof}
presents the stratification technique used in the analysis and
applies it to derive both local and global properties of
minimizers. The proof of a technical lemma is postponed to the
Appendix. Finally, Section~\ref{sec: main} contains the proof of the 
main results. Specifically,  \Cref{thm:crystallization} is
proved in Section~\ref{sec:thm:crystallization} and \Cref{thm:fluctuation}  is
proved in Section~\ref{sec:fluctuation}.

\section{Setting and main results}\label{sec:setting}

We consider particle systems in two dimensions, and model their
interaction by classical potentials in the frame of Molecular
Mechanics~\cite{Molecular, Lewars}. Let $\beta >0$ be given and define the \emph{substrate} by $\mathcal{L}^-=\mathbb{Z}^2 \cap \{x_2 \leq 0\}$. Indicating  a   \emph{configuration of particles} by $C_N =\{x_1,\ldots,x_N\} \subset \mathbb{R}^2\cap \{x_2 >0\}$, we define its energy by
\begin{align}\label{eq: main energy}
\mathcal{F}_\beta(C_N) = \frac{1}{2}\underset{x_i \neq x_j}{\sum_{x_i,x_j \in C_N}} v_2(|x_i-x_j|) + 
\beta \sum_{x_i \in C_N,  z \in \mathcal{L}^-} v_2(|x_i-z|) + \frac{1}{2} \sum_{x,y,z } v_3(\theta_{x,y,z})\,,
\end{align}
where $v_2$ and $v_3$ denote two-body and three-body interaction
potentials, respectively, which are specified below and depicted in
Figure \ref{fig:potentials}.  The third sum runs over triples
$(x,y,z)\in ( C_N \cup \mathcal{L}^-)^3$ such that $|x-y|,|y-z|<r_0$
(to be defined later, see ($\rm ii_2$)) and $\theta_{x,y,z}$ denotes
the angle formed by the vectors $x-y$ and $z-y$ (counted clockwisely).
The factor $\frac{1}{2}$ accounts for double counting of bonds and
angles. We fix  $0<\varepsilon<\varepsilon_0$ for $\varepsilon_0 <
\frac{\pi}{6}$ specified in \Cref{lemma:elementaryprop} and
\Cref{rem: eps}.  The \emph{two-body potential} $v_2 \colon
[0,+\infty) \to {\mathbb{R}}  \cup \lbrace + \infty \rbrace $  is
asked to satisfy  
\begin{itemize}
\item[($\rm i_2$)] $\min_{r \geq 0} v_2(r) = v_2(1)=-1$ and $v_2(r) >-1$ if $r \neq 1$;
\item[($\rm ii_2$)] There exists $1<r_0<\sqrt{2}$ such that $v_2(r)=0$ for all $r \geq r_0$;
\item[($\rm iii_2$)] For all $r \in [0,1-\varepsilon] $  it holds that   $v_2(r) > \varepsilon^{-1}$.
\end{itemize}
The \emph{three-body potential} $v_3 \colon [0,2\pi] \to \mathbb{R}$
 is asked to satisfy 
\begin{itemize}
\item[($\rm i_3$)] $v_3(\theta)=v_3(2\pi-\theta)$ for all $\theta \in [0,2\pi]$;
\item[($\rm ii_3$)] $v_3(k\pi/2)=0$ for $k=1,2,3$ and $v_3(\theta) >0$ if $\theta \notin \{\pi/2,\pi,3\pi/2\}$;
\item[($\rm iii_3$)] $v_3(\theta)  \geq   4(\pi/10 -\varepsilon)^{-1} \, |\theta -\pi |$ for all $\theta \in [\pi- \eps,\pi + \eps]$     with equality only if $\theta = \pi$;  
\item[($\rm iv_3$)] If $\theta \notin [\pi/2-\varepsilon,\pi/2+\varepsilon] \cup [\pi-\varepsilon,\pi+\varepsilon] \cup [3\pi/2-\varepsilon,3\pi/2+\varepsilon]$, then 
\begin{align*}
 v_3(\theta) >  \max\{1, 2  \beta\}  \frac{4}{(1-\varepsilon)^2}\Big(\sqrt{2} +\frac{1}{2}\Big)^2\,.
\end{align*} 
\end{itemize}
We briefly comment on the assumptions. Condition ($\rm i_2$) on a
unique minimum (here normalized to $1$)   is natural, e.g., it is
valid for Lennard-Jones-type potentials.  Assumption ($\rm ii_2$)
states that $v_2$ has compact support. In particular, it ensures that
for configurations $ C_N  \subset \mathbb{Z}^2$ only  particles  at distance $1$ interact, which are  usually  referred to as \emph{nearest neighbors} in the literature. Eventually, ($\rm iii_2$) prevents clustering of points  in the sense that pairs of  particles  have distance at least $1-\eps$, see    \Cref{lemma:elementaryprop} below for details. 

  Condition ($\rm i_3$) ensures that the potential $v_3$ does not depend on how (clockwise or counter-clockwise) bond angles are measured, and ($\rm ii_3$) guarantees  that for $ C_N \subset \mathbb{Z}^2$ there is no contribution stemming from the three-body interaction. Slope conditions similar to   ($\rm iii_3$) have been used in~\cite{FriedrichKreutzHexagonal,FriedrichKreutzSquare,FriedrichKreutz, Mainini-Piovano, Mainini} in order to obtain crystallization on the square or hexagonal lattice. As a consequence,  the potential is necessarily nonsmooth  at $\pi$.  Let us mention that in the  other works (except for~\cite{FriedrichKreutz}) the condition is needed at  \emph{all} minimum points of $v_3$, whereas here it is only required at $\pi$.  
  
  We also point out that in this work the focus lies on devising a general proof strategy for crystallization relative to a substrate,  and thus all appearing  specific numerical constants   are chosen for computational simplicity rather than optimality.  The shape of the two potentials $v_2$ and $v_3$ is illustrated in Figure~\ref{fig:potentials}.

\begin{figure}[htp]
\begin{tikzpicture}
\tikzset{>={Latex[width=1mm,length=1mm]}};

\draw[dash pattern=on 1.5pt off 1pt,ultra thin](1,-.04)--(1,-.475);

\draw[->](0,-1)--++(0,4) node[anchor= east] {$v_{2}(r)$};
\draw[->](-1,0)--++(4.5,0) node[anchor =north] {$r$};

\draw(2.6,-.04)--++(0,.08) node[anchor=south]{$r_0$}; 

\draw(1,-.04)--++(0,.08);

\draw(1,0) node[anchor=south]{$1$}; 

\draw[dash pattern=on 1.5pt off 1pt,ultra thin](0,-.475)--(1,-.475);

\draw(0,-.475) node[anchor=east]{$-1$};

\draw[thick,domain=.725:2.5, smooth, variable=\x] plot ({\x}, {.025+.5/(\x*\x*\x*\x*\x*\x*\x*\x)-1/(\x*\x*\x*\x)});
\draw[thick](2.49,0)--++(0:.8);

\begin{scope}[shift={(7,0)}]

\draw[dash pattern=on 1.5pt off 1pt,ultra thin](2,0)--++(110:.5);

\draw[dash pattern=on 1.5pt off 1pt,ultra thin](2,0)--++(110:-.25);
\draw[->](0,-1)--++(0,4) node[anchor= east] {$v_{3}(\theta)$};
\draw[->](-1,0)--++(6,0) node[anchor =north] {$\theta$};

\draw(1,-.04)--++(0,.08); 
\draw(3,-.04)--++(0,.08); 

\draw(1,0) node[anchor=north]{$\pi/2$};

\draw(3,0) node[anchor=north]{$3\pi/2$};

\draw(4,-.04)--++(0,.08); 
\draw(2,-.04)--++(0,.08); 
\draw(4,0) node[anchor=north]{$2\pi$};
\draw(2,-.06) node[anchor=north]{$\pi$};

\draw[thick,domain=2.93333:3.04196, smooth, variable=\x] plot ({\x}, {70*(\x-3)*(\x-3)});

\draw[thick,domain=3.04196:4, smooth, variable=\x] plot ({\x}, {1.5*(-.21 - (\x-2.9)*(\x-5.1))});

\draw[thick,domain=0:0.958042, smooth, variable=\x] plot ({\x}, {1.5*(-.21 - (\x-1.1)*(\x+1.1))});

\draw[thick,domain=2:2.93333, smooth, variable=\x] plot ({\x}, {5*(-.11 - (\x-1.9)*(\x-3.1))});

\draw[thick,domain=0.958042:1.06667, smooth, variable=\x] plot ({\x}, {70*(\x-1)*(\x-1)});

\draw[thick,domain=1.06667:2, smooth, variable=\x] plot ({\x}, {5*(-.11 - (\x-.9)*(\x-2.1))});
\end{scope}
\end{tikzpicture}
\caption{The potentials $v_2$ and $v_3$.}
\label{fig:potentials}
\end{figure}
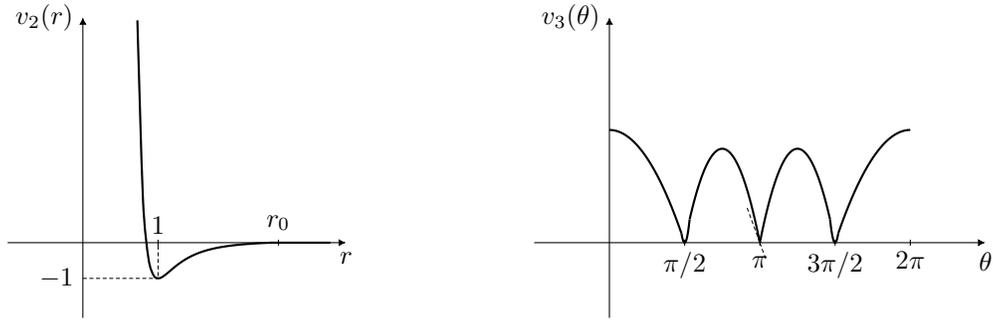

 Given $\beta >0 $ and $N \in \mathbb{N}$, it will be convenient to
 consider the \emph{normalized energy}  $ 2  (\mathcal{F}_\beta(C_N)
 +2N)$ which essentially counts the number of missing bonds at the
  free  surface of the configuration and the number of bonds
 between configuration and substrate  (plus the angle part, if
 $C_N \not \subset \Z^2$).  In this regard, we further define the \emph{normalized  minimal energy} 
\begin{align}\label{def:mbetan}
m_\beta(N) = \min_{h \in \{1,\ldots,N\}} \left(2h +  2(1-\beta) \left\lceil
\frac{N}{h} \right\rceil \right)
\end{align}
 and the largest \emph{optimal height} by 
\begin{align}\label{hstern}
 h_*(\beta,N) =\max\left\{h  \in \{1,\ldots,N\}  \colon 2h + 2(1-\beta) \left\lceil
\frac{N}{h} \right\rceil =m_\beta(N) \right\}\,.
\end{align}
We now  state the main results of the paper.   

\begin{theorem}[Crystallization]\label{thm:crystallization}
For all $\beta >0  $ there exists $N_\beta \in \mathbb{N}$ such that for each $C_N \in (\R^2)^N$ with $N \geq N_\beta$  it holds that 
 \begin{align}\label{eq: main claim0}
 \mathcal{F}_\beta(C_N) \geq  -2N + \frac{1}{2}   m_\beta(N)\,.
\end{align}
In case of equality, we have  $C_N \subset \mathbb{Z}^2 \cap
\{x_2>0\}$  and  $\#\left(C_N \cap \{x_2=1\}\right)\geq N/h_*(\beta,N) $.  
\end{theorem}
A configuration of minimal energy is depicted in Figure~\ref{fig:groundstates}.

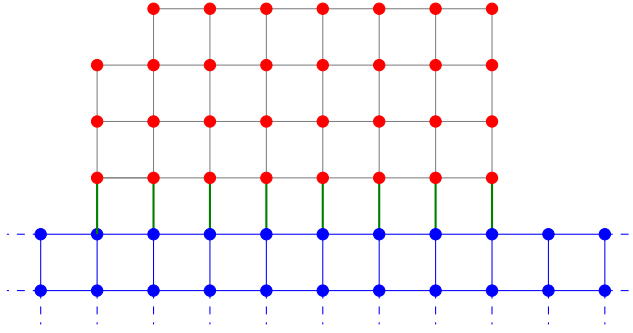
\begin{figure}[htp]
\begin{tikzpicture}[scale=0.75]

\foreach \j in {1,2,3}{
\draw[gray](1,\j)--++(0,-1);
\draw[gray](1,\j)--++(1,0);
}

\foreach \i in {2,...,8}{
\draw[gray](\i,1)--++(0,3);
}

\foreach \i in {2,...,8}{
\draw[thick, green!50!black](\i,1)--++(0,-1);
}

\draw[blue](0,0)--++(10,0);
\draw[blue](0,-1)--++(10,0);
\draw[blue,dashed](0,0)--++(-.6,0);
\draw[blue,dashed](0,-1)--++(-.6,0);
\draw[blue,dashed](10,0)--++(.6,0);
\draw[,blue,dashed](10,-1)--++(.6,0);

\foreach \i in {0,...,10}{
\draw[blue](\i,0)--++(0,-1);
\draw[blue,dashed](\i,-1)--++(0,-.6);

}

\foreach \j in {0,-1}{
\foreach \i in {0,...,10}{
\draw[fill=blue,blue](\i,\j) circle(.1);

}
}

\draw[gray](2,1)--++(-1,0)++(0,-1);
\draw[thick, green!50!black](2,1)++(-1,0)--++(0,-1);

\foreach \j in {1,...,4}{
\draw[gray](2,\j)--++(6,0);
}

\draw[fill=red,red](1,1) circle(.1);
\draw[fill=red,red](1,2) circle(.1);
\draw[fill=red,red](1,3) circle(.1);

\foreach \j in {1,...,4}{
\foreach \i in {2,...,8}{
\draw[fill=red,red](\i,\j) circle(.1);

}
}

\end{tikzpicture}
\caption{A configuration of minimal energy for  $N=31$ and $\beta=3/7$. The  bold  green bonds   indicate the interactions between the substrate and the crystal and the gray bonds indicate the interaction between crystal  particles.  }
\label{fig:groundstates}
\end{figure}

\begin{remark}[Finer characterization of minimizers]\label{rem:convexity etc.} \rm
Given $N \geq N_\beta$, minimizers  $C_N$  of $\mathcal{F}_\beta$ satisfy some more geometric properties:
\begin{itemize}
\item[{\rm(i)}] $C_N$ is {\it convex by rows and columns}, i.e., if for some $k \in \mathbb{N}$ and $j \in \{1,2\}$ we have that  $x, x+ke_j \in C_N$,  then $x+m e_j \in C_N$ for all $m=0,\ldots,k$.
\item[{\rm (ii)}] Every {\it column} of $C_N$ is either empty or of the form $C_N \cap (\{k\} \times \mathbb{Z}) = \{1,\ldots,\alpha_k\}$ for some $\alpha_k \in \mathbb{N}$.
\item[{\rm (iii)}] Setting $l= \#\{k\colon (\{k\} \times \mathbb{Z})
  \cap C_N   \neq \emptyset\}$ and $h =\#\{k\colon (\mathbb{Z}\times
  \{k\})  \cap C_N   \neq \emptyset\}$ we have that 
\begin{align*}
\mathcal{F}_\beta(C_N) =  - 2N + \frac{1}{2} \big(  2(1-\beta)  l +2h
  \big)\,.   
\end{align*}
\end{itemize}
\end{remark}

\begin{remark}[$\beta \geq 1$: Wetting regime]\label{rem:wetting} \rm
If $\beta \ge 1$, it is elementary to check that $ h_*(\beta,N) =
1$. Thus,  \Cref{thm:crystallization} and \Cref{rem:convexity etc.}  show that in this case  minimizers are  horizontal   integer shifts of the single chain of  particles   $\lbrace (i,1) \colon i = 1,\ldots,N \rbrace$, and the minimal energy is given by  $ -(N-1)  - \beta N$.
\end{remark}

In order to state the next theorem, we introduce a \emph{reference rectangle}
\begin{align*}
R(\beta,N) =  \left\{1,\ldots,\left\lceil \frac{N}{h_*(\beta,N)}\right\rceil\right\} \times \{1,\ldots, h_*(\beta,N)\} \,.
\end{align*}

\begin{theorem}[Fluctuations]\label{thm:fluctuation}
 Let $\beta \in  (0,1)  $ and $C_N  \subset \Z^2 \cap \{x_2>0\}$
   be a minimizer of $\mathcal{F}_\beta$. Then the following upper bounds  hold true:
\begin{itemize}
\item[{\rm ($U_1$)}] If $\beta \in \mathbb{Q}$, then there exists $C_\beta>0$ such that (up to a horizontal translation  in $\mathbb{Z} \times \{0\}$) 
\begin{align*}
\#\left(C_N   \triangle  R(\beta,N)  \right) \leq C_\beta
  N^{3/4}\,.
\end{align*}
% for all horizontal translations $\tau \in \mathbb{Z} \times \{0\}$. 
\item[{\rm ($U_2$)}] If $\beta \in \mathbb{R}\setminus \mathbb{Q}$ is algebraic, then for every $ \delta >0$ there exists $C_{\beta, \delta }>0$ such that  (up to a horizontal translation in $\mathbb{Z} \times \{0\}$) 
\begin{align*}
\#\left(C_N  \triangle   R(\beta,N)  \right) \leq C_{\beta, \delta } N^{1/3+ \delta }\,.
\end{align*}
% for all horizontal translations $\tau \in \mathbb{Z} \times \{0\}$. 
\end{itemize}
Moreover,  the following lower  bounds hold true:
\begin{itemize}
\item[{\rm ($L_1$)}] If $\beta \in \mathbb{Q}$,  there exists
  $c_\beta>0$ and a sequence $\{ N_k \}_k \subset \mathbb{N}$ with
  $N_k \to \infty$ as $k \to \infty$ such that  for all $k \in
  \mathbb{N}$ there exists  a minimizer of $\mathcal{F}_\beta$,
  denoted by $C_{N_k}$, % and a horizontal translation$\tau \in
  % \mathbb{Z} \times \{0\}$ 
  satisfying    
\begin{align*}
\#\left(C_{N_k}  \triangle \left(R(\beta,N_k) +\tau\right)\right) \geq c_\beta N_k^{3/4}\,
\end{align*}
 for all horizontal translations $\tau \in \mathbb{Z} \times \{0\}$. 
\item[{\rm ($L_2$)}] If $\beta \in \mathbb{R}\setminus \mathbb{Q}$ is
  algebraic, then for all $ \delta >0$ there exists a constant
  $c_{\beta, \delta }>0$ and a sequence $\{  N_k \}_k \subset
  \mathbb{N}$ with $N_k \to \infty$ as $k \to \infty$ such that  for
  all $k \in \mathbb{N}$ there exists  a minimizer of
  $\mathcal{F}_\beta$, denoted by $C_{N_k}$, %  and a horizontal
                                % translation  $\tau \in \mathbb{Z}
                                % \times \{0\}$ 
  satisfying    
\begin{align*}
\#\left(C_{N_k}  \triangle \left(R(\beta,N_k) +\tau\right)\right) \geq c_{\beta, \delta } N_k^{1/3- \delta }\,
\end{align*}
 for all horizontal translations $\tau \in \mathbb{Z} \times \{0\}$.
\end{itemize}
\end{theorem}
 The result provides a fluctuation law of $N^{3/4}$ in the case of rational $\beta$ and a law of $N^{1/3}$ for irrational, algebraic $\beta$. The case of a transcendental interaction parameter is not addressed here. Some configurations satisfying the optimal fluctuation estimates are depicted in Figure~\ref{Fig:fluctuations}. 

\begin{figure}[htp]
\begin{tikzpicture}

\tikzset{>={Latex[width=1mm,length=1mm]}};

\begin{scope}[shift={(7,-4)}]

\draw[ultra thin, gray!10!white, fill=gray!10!white](0,0) rectangle(6,-1);

\draw[thin, gray!50!white, fill=gray!10!white, fill opacity=.5](.5,0)--(5.5,0)--(5.5,1.5)--(5.4,1.5)--(5.4,2)--(.5,2)--(.5,0);

\draw(.5,0)--(5.5,0)--(5.5,2)--(1,2)--(1,1.9)--(.5,1.9)--(.5,0);

\draw[<->](.5,2.2)--++(.5,0);
\draw(.8,2)++(0,.25) node[anchor=south]{$\sim N^{1/3}$};

\end{scope}

\begin{scope}[shift={(0,-4)}]

\draw[ultra thin, gray!10!white, fill=gray!10!white](0,0) rectangle(6,-1);

\draw(.5,0)--(5.5,0)--(5.5,1.5)--(5.4,1.5)--(5.4,2)--(.5,2)--(.5,0);

\draw[<->](5.8,1.5)--++(0,.5);

\draw(5.8,1.5)++(0,.25) node[anchor=west]{$\sim N^{1/3}$};

\end{scope}

\begin{scope}[shift={(0,0)}]

\draw[ultra thin, gray!10!white, fill=gray!10!white](0,0) rectangle(6,-1);

\draw(.5,0)--(5,0)--(5,1.2)--(4.9,1.2)--(4.9,2)--(.5,2)--(.5,0);

\draw[<->](5.3,1.2)--++(0,.8);

\draw(5.3,1.2)++(0,.4) node[anchor=west]{$\sim N^{1/2}$};

\end{scope}

\begin{scope}[shift={(7,0)}]

\draw[ultra thin, gray!10!white, fill=gray!10!white](0,0) rectangle(6,-1);

\draw[ultra thin, gray!50!white, fill=gray!10!white, fill opacity=.5](.5,0)--(5,0)--(5,1.2)--(4.9,1.2)--(4.9,2)--(.5,2)--(.5,0);

\draw(.5,0)--(5.6,0)--(5.6,1.6)--(.5,1.6)--(.5,0);

\draw[<->](5,-.2)--(5.6,-.2);
\draw(5.3,-.2) node[anchor=north]{$\sim N^{1/4}$};

\end{scope}

\end{tikzpicture}
\caption{Schematic depiction of ground states that satisfy the optimal fluctuation estimates. The above two for $\beta \in \mathbb{Q}$ and the bottom two for $\beta \in \mathbb{R}\setminus \mathbb{Q}$ algebraic.}
\label{Fig:fluctuations}
\end{figure}
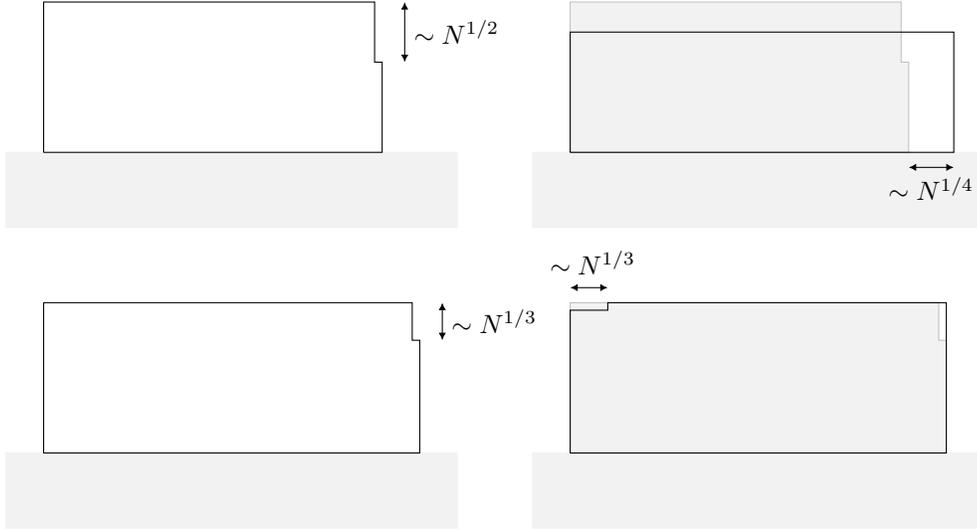

A direct consequence of the sharp fluctuation estimates of
\Cref{thm:fluctuation} is the characterization of the Winterbottom shape
emerging in the large-particle limit $N \to \infty$. Arguing along the
lines of \cite{Yuen} and using also \Cref{lem:quadratic} below, one may prove the
following.

\begin{corollary}[Winterbottom shape]
 Let $\beta \in  (0,1)  $.  The optimal height $ h_*(\beta,N)$ defined in \eqref{hstern} satisfies $N^{-1/2}h_*(\beta,N) \to \sqrt{1-\beta}$ as $N \to \infty$. Letting    $C_N =\{x_1^N,
 \dots, x_N^N\}
 \subset \Z^2 \cap \{x_2>0\}$ be a minimizer of $\mathcal{F}_\beta$
 with $\min\{x_1^j :  j=1,\dots,N\}=1$, and  defining  the sequence of rescaled
 empirical measures
 $$\mu_N := \frac{1}{N}\sum_{j=1}^N  \delta_{x^N_j/\sqrt{N}}\,,$$ 
as $N \to \infty$, the measures $\mu_N$ converge weakly* to the
restriction of the Lebesgue in $\R^2$ to the rectangle
$$R=\left[0,\frac{1}{\sqrt{1-\beta}}\right]\times [0,\sqrt{1-\beta}]\,.$$
\end{corollary}

\section{Stratification}\label{sec:mainproof}

After a short preliminary on graph theory, this section is devoted to  detail   the main technique of this paper, namely a modification of bond graphs, called \emph{stratification}.

\subsection{Bond graph} \label{subsec:graph theory}   Let  $G=(V,E)$
be a graph, where $V$ indicates the set of \emph{vertices} and $E
\subset \{\{x,y\} \colon x,y \in V \text{ and } x\neq y \}$ is  the
set of \emph{edges}.  For our purposes, we shall always assume
that graphs are embedded in $\R^2$.  % It
                                % is obvious that  $(V,E)$ is a graph
                                % embedded  in $\mathbb{R}^2$.
For $x \in V$, we denote the  (graph)  neighborhood with respect to $G$ and the  (graph)  neighborhood with respect to $\mathcal{L}^-$ by
\begin{align*}
%\label{def:graph-neighborhood}
\mathcal{N}(x,E) :=\{y \in V \colon \{x,y\} \in E\}\,, \quad \mathcal{N}_{\mathcal{L}^-}(x) :=\{y \in \mathcal{L}^- \colon |x-y|\leq r_0\}\,,
\end{align*}
  for $r_0>0$ as given in {\rm ($\rm{ii}_2$)}.  For all such graphs $G=(V,E)$ we define
\begin{align} 
\label{def:F}
F_\beta(G) = F_{\mathrm{bond},\beta}(G) + F_{\mathrm{ex},\beta}(G)\,,
\end{align}
where
\begin{align}
\label{def:Fbond}
F_{\mathrm{bond},\beta}(G) = \sum_{x \in V} \big(4-\#\mathcal{N}(x,E)- 2 \beta\#\mathcal{N}_{\mathcal{L}^-}(x) \big)\,
\end{align}
 is the \textit{bond energy} and  
\begin{align*}
%\label{def:Felastic}
 F_{\mathrm{ex},\beta}(G)  =  \underset{x,y\in V}{\sum_{\{x,y\} \in E}} (v_2(|x-y|)+1) +  2  \beta  \sum_{x\in V,y\in \mathcal{N}_{\mathcal{L}^-}(x) } (v_2(|x-y|)+1)+ \sum_{ x,y,z} v_3(\theta_{x,y,z})
\end{align*}
is the \textit{excess energy}. For $V'\subset V$,  let 
$G[V']$  be   the (vertex) induced subgraph of $V'$ in $G$, that is $G[V']= (V',E')$ with $E'= \{\{x,y\} \in E\colon x,y \in V'\}$. 
% We can localize the excess energy by 
%%we also define the localized excess energy by
%\begin{align}\label{def:Felasticlocal}
%   F_{\mathrm{ex},\beta}(V') =
%  F_{\mathrm{ex},\beta}(G[V'])\,.
%\end{align}
 We will identify each configuration $C_N \subset \mathbb{R}^2$ with its \emph{natural bond graph} $G_{\mathrm{nat}}=(V,E_{\mathrm{nat}})$, where $V=C_N$ and the \emph{natural edges} are given by
\begin{align}\label{eq: relation-new}
  E_\mathrm{nat}   = \{\{x,y\} \colon x,y \in C_N, |x-y|\leq r_0\}\,.
\end{align}
 This definition is motivated by the  relation to the energy $\mathcal{F}_\beta$ defined in  \eqref{eq: main energy}, namely 
\begin{align}\label{eq: relation}
2\mathcal{F}_\beta(C_N) =  -4N  + F_\beta(G_{\mathrm{nat}})\,.
\end{align} 
  In Section~\ref{subsec:stratbondgraph} below,  we will successively modify $E_{\mathrm{nat}}$ to a smaller set of edges $E\subset E_\mathrm{nat}$, according to a set of given rules.

\begin{definition}\label{def:epsregular} We say that $G=(V,E)$ is \emph{$\varepsilon$-regular} for $\varepsilon >0$ small if  the following two  conditions hold:
\begin{itemize}
\item[(i)] For $x,y \in V  \cup\mathcal{L}^-$ with $x\neq y$  it holds that
\begin{align*}
 |x-y| \ge 1-\varepsilon\,;
\end{align*}
\item[(ii)] For each bond angle $\theta=\theta_{x,y,z}$ with $x,y,z \in  V  \cup \mathcal{L}^-$ it holds that 
\begin{align*}
\theta \in [\pi/2-\varepsilon,\pi/2+\varepsilon] \cup [\pi-\varepsilon,\pi+\varepsilon] \cup [3\pi/2-\varepsilon,3\pi/2+\varepsilon]\,.
\end{align*}
\end{itemize}
\end{definition}
Note that, if $G_{\mathrm{nat}}=(V,E_{\mathrm{nat}})$ is $\varepsilon$-regular, then  $G=(V,E)$ is $\varepsilon$-regular for all $E\subset E_{\mathrm{nat}}$.

\begin{lemma}\label{lemma:elementaryprop} There exists $\varepsilon_0 >0$ such that the following holds true:  if $v_2$, $v_3$ satisfy {\rm ($\rm{i}_2$)}--{\rm ($\rm{iii}_2$)} and {\rm ($\rm{i}_3$)}--{\rm ($\rm{iv}_3$)} for some $0<\varepsilon<\varepsilon_0$ and  if $C_N$ is a minimizer of \eqref{eq: main energy}, then its natural bond graph $G_{\mathrm{nat}}=(V,E_{\mathrm{nat}})$ is $\varepsilon$-regular. Moreover, it holds that  $\#\mathcal{N}(x,E_{\mathrm{nat}}) + \#\mathcal{N}_{\mathcal{L}^-}(x) \leq 4$ and $ \#\mathcal{N}_{\mathcal{L}^-}(x) \leq 1$  for all $x \in V$.
\end{lemma}

The proof of this statement is similar to the one of~\cite[Lemma~3.2]{FriedrichKreutz}. We include it in   Appendix~\ref{appendix} for convenience of the reader.     For the remainder of this paper, we assume that $\varepsilon_0>0$ is chosen small enough such that \Cref{lemma:elementaryprop} holds true and that $v_2$ and $v_3$ satisfy $({\rm i}_2)$--$({\rm iii}_2)$ and $({\rm i}_3)$--$({\rm iv}_3)$  for some $0<\varepsilon<\varepsilon_0$.

\subsection{Stratified bond graph}  \label{subsec:stratbondgraph}

Given $G=(V,E)$, we say that $\gamma=(x_1,\ldots,x_n)$ with $x_i \in V   $   for all $i=1,\ldots,n$ is  a \emph{straight path} if  $n \ge 2$ and the following holds:  
\begin{itemize}
\item[(i)] $\{x_i,x_{i+1}\} \in E$ for all $i=1,\ldots,n-1\,$;
\item[(ii)] $\theta_i \in [\pi-\varepsilon,\pi+\varepsilon]$ for all $i \in 2,\ldots,n-1$, where   $\theta_i =\theta_{x_{i+1},x_i,x_{i-1}}$;   
\item[(iii)] $\{x_i,x_{i+1}\} \neq \{x_{j},x_{j+1}\}$ for all $i,j =1,\ldots,n-1$, $ j \neq i$.
\end{itemize}
 (If $n=2$, (ii) and (iii) are empty.)  The set of straight paths is denoted by
\begin{align*}
%\label{def:gamma}
\Gamma (G) :=\{\gamma \text{ straight path}\}\,.
\end{align*}
 We drop $G$ and  simply  write $\Gamma$ if no confusion arises.
If $\gamma \in \Gamma$ and $x_1=x_n$, we say that $\gamma$ is \emph{closed} and otherwise that $\gamma$ is \emph{open}.   We define   
\begin{align*}
%\label{def:V0V1}
\begin{split}
&V_i :=\{x \in V \colon \# \mathcal{N}(x,E) = i  \} \text{  for $i=0,\ldots,4$\,,} \\&V_2^\pi :=\{x \in V_2 \colon  \theta_{x_1,x,x_2},   \in [\pi-\varepsilon,\pi+\varepsilon] \text{ where } \mathcal{N}(x,E) = \lbrace x_1,x_2 \rbrace \}\,.
\end{split}
\end{align*}
 Note that in the second definition one could equally use the angle $\theta_{x_2,x,x_1}$ as $\theta_{x_2,x,x_1} = 2\pi - \theta_{x_1,x,x_2}$. We define the \emph{set of strata} by
\begin{align*}
%\label{def:strata}
\mathcal{S}(G):=  \mathcal{S}_\Gamma \cup \bigcup_{x \in V_0 \cup V_1 \cup V_2^\pi} s(x), \quad \text{ where } \mathcal{S}_\Gamma:=\{\gamma \in \Gamma \colon \gamma \text{ is a maximal element  w.r.t.}  \subseteq \}\,.
\end{align*}
Here, we set $s(x) =\{(x),(x)\}$ for $x \in V_0$ and  $s(x) = \{(x)\}$ for  $x \in V_1 \cup V_2^\pi$.  Adding the \emph{degenerate stratum} $(x)$   twice for $V_0$ and once for $V_1 \cup V_2^\pi$ has no geometrical interpretation but is merely convenient to relate  the number of  strata to $F_{{\rm bond},\beta}$, see \Cref{lemma:opengraph} below. In particular,  it will ensure that  each  particle  is contained in exactly two strata, see \Cref{lemma:opengraph}(iii).

  We say that $s\in \mathcal{S}_\Gamma$,  $s= (x_1,\ldots,x_n)$,  is an \emph{interacting stratum} (with the substrate) if there is $z_0 \in \mathcal{N}_{\mathcal{L}^-}(x_1)$ such that $\theta_{z_0,x_1,x_2} \in [\pi-\varepsilon,\pi+\varepsilon]$ and/or $z_{n+1} \in \mathcal{N}_{\mathcal{L}^-}(x_n)$ such that  $\theta_{x_{n-1},x_n,z_{n+1}} \in [\pi-\varepsilon,\pi+\varepsilon]$. The \emph{set of interacting strata}   is defined by 
 \begin{align*}
  \mathcal{S}_{\rm int} (G)   \defas \{ \gamma \in \mathcal{S}_\Gamma \colon & \gamma \text{ interacting stratum}\} \cup \{(x) \colon x\in   V_0    \cup V_2^\pi, \mathcal{N}_{\mathcal{L}^-}(x) \neq \emptyset\}\\&\cup \{(x) \colon x \in V_1, \mathcal{N}_{\mathcal{L}^-}(x) \neq \emptyset,  x\notin \gamma \text{ for each interacting stratum } 
 \gamma  \} \,.
 \end{align*}
  Also here the addition of degenerate strata is  convenient to relate  the number of  strata to $F_{{\rm bond},\beta}$.  
The \emph{set of noninteracting strata} is given by $ \mathcal{S}_{\rm no}  (G)  :=\mathcal{S} (G)   \setminus \mathcal{S}_{\rm int} (G) $. We say that $s \in \mathcal{S}_{\rm int} (G)$, $s= (x_1,\ldots,x_n)$, is \emph{not double touching} if either  $\mathcal{N}_{\mathcal{L}^-}(x_1) \neq \emptyset$ or  $\mathcal{N}_{\mathcal{L}^-}(x_n)\neq \emptyset$. We drop $G$ and write $\mathcal{S}$, $\mathcal{S}_{\rm int}$,  and  $\mathcal{S}_{\rm no}$ if no confusion arises. Different possibilities of interacting strata are illustrated in Figure~\ref{fig:interacting-strata}.
\begin{figure}[htp]
\begin{tikzpicture}[scale=0.75]

\draw[ultra thin,gray](0:14)--++(100:1)--++(110:1)--++(120:1)--++(130:1)--++(140:1)--++(150:1)--++(160:1)--++(170:1)--++(180:1)--++(190:1)--++(200:1)--++(210:1)--++(225:1)--++(235:1)--++(245:1)--++(255:1)--(3,0);

\draw[fill=red,red](0:14)++(100:1)circle(.05)++(110:1)circle(.05)++(120:1)circle(.05)++(130:1)circle(.05)++(140:1)circle(.05)++(150:1)circle(.05)++(160:1)circle(.05)++(170:1)circle(.05)++(180:1)circle(.05)++(190:1)circle(.05)++(200:1)circle(.05)++(210:1)circle(.05)++(225:1)circle(.05)++(235:1)circle(.05)++(245:1)circle(.05)++(255:1)circle(.05);

\draw[ultra thin,gray](5,1)--++(2,0);
\draw[ultra thin,gray](6,1)--++(0,-1);
\draw[fill=red, red](6,1) circle(.05);

\draw[ultra thin,gray](7,1)--++(0,-1);
\draw[fill=red,red](7,1) circle(.05);

\draw[ultra thin,gray](5,1)--++(0,-1);
\draw[fill=red,red](5,1) circle(.05);

\draw[ultra thin,gray](12,1)--++(0,-1);
\draw[fill=red,red](12,1) circle(.05);

\draw[fill=truegreen,truegreen](12,1) circle(.05);
\draw[fill=truegreen,truegreen](5,1) circle(.05);
\draw[fill=truegreen,truegreen](6,1) circle(.05);
\draw[fill=truegreen,truegreen](7,1) circle(.05);

\draw(4.6,0.9)++(0,.4) node[anchor=west]{$x_1$};
\draw(5.6,0.9)++(0,.4) node[anchor=west]{$x_2$};
\draw(6.6,0.9)++(0,.4) node[anchor=west]{$x_3$};
\draw(11.6,0.9)++(0,.4) node[anchor=west]{$x_4$};

\draw[ultra thin, gray](0,0)  grid(14,-1);

\draw[ultra thin, gray, dashed](-.5,0)--++(15,0);
\draw[ultra thin, gray, dashed](-.5,-1)--++(15,0);
\foreach \j in {0,...,14}{
\draw[ultra thin, gray,dashed](\j,-1)--++(0,-.5);
}

\foreach \j in {1,...,4}{
\draw[ultra thin,gray](0,\j)--++(0,-1);
\draw[fill=red,red](0,\j) circle(.05);

}

\foreach \j in {0,...,14}{
\foreach \i in {-1,0}{
\draw[fill=blue,blue](\j,\i) circle(.05);
}
}

\end{tikzpicture}
\caption{Different interacting strata. Four degenerate ones are indicated in green, where $x_1,x_3 \in V_1$, $x_2 \in V_2^\pi$, and $x_4 \in V_0$.}
\label{fig:interacting-strata}
\end{figure}
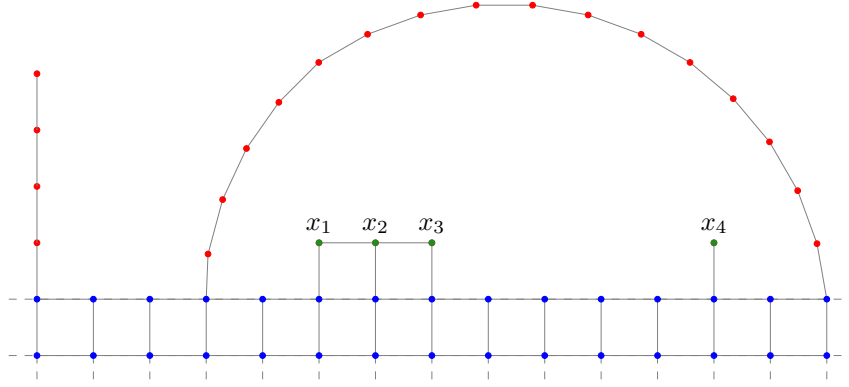

  \begin{definition}[Length, orthogonal strata, and span]\label{def: strata-ort} Let $s \in \mathcal{S}$.
  By $l(s):= \# s$ we denote its \emph{length}. We define the \emph{set of orthogonal strata} to $s$ by
\begin{align*}
%\label{def:verticalstrata}
\mathcal{S}^\perp(s)=\{s' \in \mathcal{S} \setminus \{s\} \colon s\cap s'\neq \emptyset\}\,.
\end{align*}
Given $s_0 \in \mathcal{S}$ we define
\begin{align}\label{def:span}
\mathrm{span}(s_0) = \bigcup_{s \in \mathcal{S}^\perp(s_0)}s\,.
\end{align}
\end{definition}
 A stratum $s \in \mathcal{S}$ and its orthogonal strata are illustrated in Figure~\ref{fig:orthogonal}.  We proceed with the definition of the angle excess for straight paths and a corresponding lemma.   
\begin{definition}[Angle excess] Given $\gamma=(x_1,\ldots,x_n) \in \Gamma$, we define the \emph{angle excess} by
\begin{align*}
%\label{def:angleexcess}
\theta_{\mathrm{ex}}(\gamma) := \sum^{n-1}_{i=2} |\theta_i-\pi|\,, \quad \text{where $\theta_i =\theta_{x_{i+1},x_i,x_{i-1}}$}\,. 
\end{align*}
\end{definition}

\begin{figure}
\begin{tikzpicture}[scale=.8]

\draw(0,0)++(0:1.05)++(5:1)++(-5:1)++(89.5:1)++(92:.9)++(88:.9)++(90:.25) node[anchor=south]{$s$};

\begin{scope}[shift={(3,-2.95)},rotate=90]

\draw[ultra thin,dashed](0,0)++(0:1)++(5:1)++(-5:.8)--++(90:-3) arc(180:360:.2)--++(90:6) arc (0:180:.2)--++(270:3);

\end{scope}

\draw[ultra thin,gray](0,0)++(0:1)++(5:1)++(-5:1)++(8:1)++(-2:1)++(0:1)--++(90.1:-1)--++(92:-1)--++(88:-1.1);

\draw[ultra thin,gray](0,0)++(0:1)++(5:1)++(-5:1)++(8:1)++(-2:1)--++(90.1:-1)--++(92:-1)--++(88:-1.1);

\draw[ultra thin,gray](0,0)++(0:1)++(5:1)++(-5:1)++(8:1)--++(90.1:-1)--++(92:-1)--++(88:-1.1);

\draw[ultra thin,gray](0,0)++(0:1)++(5:1)--++(89.5:-1)--++(92:-1)--++(88:-1);

\draw[ultra thin,gray](0,0)++(0:1)--++(89:-1)--++(92:-1)--++(88:-1);

\draw[ultra thin,gray](0,0)--++(89:-1)--++(92:-1)--++(88:-1);

%%%%

\draw[semithick,green!50!black](0,0)++(0:1)++(5:1)++(-5:1)++(90:-1)++(92:-1)++(88:-1)--++(180:1)--++(180:1)--++(180:1);

\draw[semithick,green!50!black](0,0)++(0:1)++(5:1)++(-5:1)++(90:-1)++(92:-1)++(88:-1)--++(0:1)--++(0:1)--++(0:1);

\draw[semithick,green!50!black](0,0)++(0:1)++(5:1)++(-5:1)++(90:-1)++(92:-1)--++(180:1)--++(180:1)--++(180:1);

\draw[semithick,green!50!black](0,0)++(0:1)++(5:1)++(-5:1)++(90:-1)++(92:-1)--++(0:1)--++(0:1)--++(0:1);

\draw[semithick,green!50!black](0,0)++(0:1)++(5:1)++(-5:1)++(90:-1)--++(180:1)--++(180:1)--++(180:1);

\draw[semithick,green!50!black](0,0)++(0:1)++(5:1)++(-5:1)++(90:-1)--++(0:1)--++(0:1)--++(0:1);

\draw[fill=black](0,0)++(0:1)++(5:1)++(-5:1)++(90:-1)++(180:1)circle(.05)++(180:1)circle(.05)++(180:1)circle(.05);

\draw[fill=black](0,0)++(0:1)++(5:1)++(-5:1)++(90:-1)++(0:1)circle(.05)++(0:1)circle(.05)++(0:1)circle(.05);

\draw[fill=black](0,0)++(0:1)++(5:1)++(-5:1)++(90:-1)++(92:-1)++(180:1)circle(.05)++(180:1)circle(.05)++(180:1)circle(.05);

\draw[fill=black](0,0)++(0:1)++(5:1)++(-5:1)++(90:-1)++(92:-1)++(0:1)circle(.05)++(0:1)circle(.05)++(0:1)circle(.05);

\draw[fill=black](0,0)++(0:1)++(5:1)++(-5:1)++(90:-1)++(92:-1)++(88:-1)++(180:1)circle(.05)++(180:1)circle(.05)++(180:1)circle(.05);

\draw[fill=black](0,0)++(0:1)++(5:1)++(-5:1)++(90:-1)++(92:-1)++(88:-1)++(0:1)circle(.05)++(0:1)circle(.05)++(0:1)circle(.05);

\draw[ultra thin,gray](0,0)++(0:1)++(5:1)++(-5:1)++(8:1)++(-2:1)--++(0:1)--++(80:.9)--++(112:.8)--++(83.5:.8);

\draw[ultra thin,gray](0,0)++(0:1)++(5:1)++(-5:1)++(8:1)++(-2:1)++(80:.9)--++(0:1);

\draw[ultra thin,gray](0,0)++(0:1)++(5:1)++(-5:1)++(8:1)++(-2:1)--++(80:.9)--++(112:.8)--++(83.5:.8);

\draw[ultra thin,gray](0,0)++(0:1)++(5:1)++(-5:1)++(8:1)--++(90:1.1)--++(90.5:.8)--++(89:.9);
\draw[ultra thin,gray](0,0)--++(89.5:1)--++(92:.8)--++(88:1);

\draw[ultra thin,gray](0,0)++(0:1)--++(89.5:1)--++(92:.9)--++(88:.9);

\draw[ultra thin,gray](0,0)++(0:1)++(5:1)--++(89.5:1)--++(92:.9)--++(88:.9);

%\draw[ultra thin,gray](0,0)++(0:1)++(5:1)++(-5:1)--++(89.5:1)--++(92:.9)--++(88:.9);

\draw[semithick,green!50!black](0,0)--++(0:1)--++(5:1)--++(-5:1)--++(8:1)--++(-2:1)--++(0:1);

\draw[semithick,green!50!black](0,0)++(0:1)++(5:1)++(-5:1)++(90:1)--++(10:1)--++(30:1)--++(0:1);

\draw[semithick,green!50!black](0,0)++(0:1)++(5:1)++(-5:1)++(90:1)++(92:1)--++(0:1);

\draw[semithick,green!50!black](0,0)++(0:1)++(5:1)++(-5:1)++(90:1)++(92:1)--++(0:-1)--++(5:-1)--++(10:-1);

\draw[semithick,green!50!black](0,0)++(0:1)++(5:1)++(-5:1)++(90:1)++(92:1)++(88:1)--++(-5:1)--++(-20:1)--++(2:1);

\draw[semithick,green!50!black](0,0)++(0:1)++(5:1)++(-5:1)++(90:1)++(92:1)++(88:1)--++(5:-1)--++(10:-1)--++(-2:-1);

\draw[fill=black](0,0)++(0:1)++(5:1)++(-5:1)++(90:1)++(92:1)++(88:1)++(5:-1)circle(.05)++(10:-1)circle(.05)++(-2:-1)circle(.05);

\draw[fill=black](0,0)++(0:1)++(5:1)++(-5:1)++(8:1)++(-2:1)++(0:1)++(80:.9)circle(.05);

\draw[fill=black](0,0)++(0:1)++(5:1)++(-5:1)++(90:1)++(92:1)++(88:1)++(-5:1)circle(.05)++(-20:1)circle(.05)++(2:1)circle(.05);

\draw[semithick,green!50!black](0,0)++(0:1)++(5:1)++(-5:1)++(90:1)--++(10:-1)--++(0:-1)--++(-5:-1);

\draw[semithick,green!50!black](0,0)++(0:1)++(5:1)++(-5:1)--++(90:1)--++(92:1)--++(88:1);

\draw[semithick,red](0,0)++(0:1)++(5:1)++(-5:1)--++(90:1)--++(92:1)--++(88:1);

\draw[semithick,red](0,0)++(0:1)++(5:1)++(-5:1)--++(90:-1)--++(92:-1)--++(88:-1);

\draw[fill=black](0,0)++(0:1)++(5:1)++(-5:1)++(90:1)++(92:1)++(0:1)circle(.05);

\draw[fill=black](0,0)++(0:1)++(5:1)++(-5:1)++(90:1)++(92:1)circle(.05)++(0:-1)circle(.05)++(5:-1)circle(.05)++(10:-1)circle(.05);

\draw[fill=black](0,0)++(0:1)++(5:1)++(-5:1)++(8:1)++(-2:1)++(80:.9)circle(.05);

\draw[fill=black](0,0)++(0:1)++(5:1)++(-5:1)++(90:1)++(10:-1)circle(.05)++(0:-1)circle(.05)++(-5:-1)circle(.05);

\draw[fill=black](0,0)++(0:1)++(5:1)++(-5:1)++(90:1)++(10:1)circle(.05)++(30:1)circle(.05)++(0:1)circle(.05);

\draw[fill=black](0,0)++(0:1)circle(.05)++(5:1)circle(.05)++(-5:1)circle(.05)++(90:1)circle(.05)++(92:1)circle(.05)++(88:1)circle(.05);

\draw[fill=black](0,0)circle(.05)++(0:1)circle(.05)++(5:1)circle(.05)++(-5:1)circle(.05)++(90:-1)circle(.05)++(92:-1)circle(.05)++(88:-1)circle(.05);

\draw[fill=black](0,0)circle(.05)++(0:1)circle(.05)++(5:1)circle(.05)++(-5:1)circle(.05)++(8:1)circle(.05)++(-2:1)circle(.05)++(0:1)circle(.05);

\draw(6.1,.15) node[anchor=west]{$s'$};

\end{tikzpicture}
\caption{The stratum $s$ in red and its orthogonal strata $\mathcal{S}^\perp(s)$ in green. One $s'\in \mathcal{S}^\perp(s)$ is encircled.}
\label{fig:orthogonal}
\end{figure}
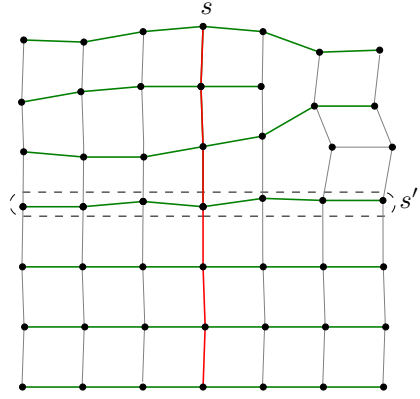

\begin{lemma}(Small angle excess)\label{lemma:excess} Let $G=(V,E)$ be an $\eps$-regular graph. The following implications hold true:
\begin{itemize}
\item[\rm (i)] If  $\displaystyle\max_{\gamma \in \Gamma} \theta_{\mathrm{ex}}(\gamma) < \frac{3\pi}{2}-\eps   $, then all $s \in \Gamma$ are open;
\item[\rm (ii)]  If $\displaystyle \max_{\gamma \in \Gamma} \theta_{\mathrm{ex}}(\gamma) < \frac{\pi}{2}-\varepsilon$, then $\#\mathcal{S}^\perp(s) = l(s)$ for all $s \in \mathcal{S}$;
\item[\rm (iii)] If $\displaystyle \max_{\gamma \in \Gamma} \theta_{\mathrm{ex}}(\gamma) < \frac{\pi}{6}-\varepsilon$, then $s_1 \cap s_2= \emptyset$ for all $s_1,s_2 \in \mathcal{S}^\perp(s)$ and for all $s \in \mathcal{S}$;
\item[\rm (iv)] If $\displaystyle \max_{\gamma \in \Gamma} \theta_{\mathrm{ex}}(\gamma) <\frac{\pi}{4}-\frac{5}{2}\varepsilon  $  and if $s\in \mathcal{S}_{\rm int}$, then $s' \in \mathcal{S}_{\rm no}$ for all $s' \in \mathcal{S}^\perp(s)$; 
\item[\rm (v)] If $\displaystyle \max_{\gamma \in \Gamma} \theta_{\mathrm{ex}}(\gamma) < \pi-4\varepsilon$, then all $s \in\mathcal{S}_{\rm int}$ are not double touching;
\item[\rm (vi)] If $\displaystyle \max_{\gamma \in \Gamma} \theta_{\mathrm{ex}}(\gamma) <\frac{\pi}{10}-\varepsilon$, if  $s_1,s_2 \in \mathcal{S}$ with $s_1 \cap s_2 \neq \emptyset$,  and if   $s' \in \mathcal{S}^\perp(s_1)$ and $s'' \in \mathcal{S}^\perp(s_2)$, then  $\hat{s}_1 \neq \hat{s}_2$ for all  $\hat{s}_1,\hat{s}_2 \in \mathcal{S}$ such that $\hat{s}_1 \in \mathcal{S}^\perp(s') $ and   $\hat{s}_2 \in \mathcal{S}^\perp(s'') $;
\item[\rm (vii)] If $\displaystyle \max_{\gamma \in \Gamma} \theta_{\mathrm{ex}}(\gamma) <\frac{\pi}{10}-  \frac{7}{5}  \varepsilon$, if  $s_1,s_2 \in \mathcal{S}$ with $s_1 \cap s_2 \neq \emptyset$,  and if there exists  $s' \in \mathcal{S}^\perp(s_1)$ such that $s'\in \mathcal{S}_\mathrm{int}$, then $s''\in \mathcal{S}_{\mathrm{no}}$ for all  $s'' \in \mathcal{S}^\perp(s_2)$.   
\end{itemize}
\end{lemma}

\begin{proof}  Statements {\rm (i)}--{\rm (iii)} follow as in~\cite[Proof of Lemma 3.6]{FriedrichKreutz}.   

We first prove {\rm (iv)}. Assume that $s=(x_1,\ldots,x_n) $ is
interacting, say at $x_1$ and  $z_1 \in 
\mathcal{N}_{\mathcal{L}^-}(x_1)$,   and assume that $s'=(y_1,\ldots,y_m)$ satisfies $s'\in \mathcal{S}^\perp(s)$, i.e., $y_j=x_l$ for some $j\in \{1,\ldots, m\}$ and some $l \in \{1,\ldots, n \}$. Assume by contradiction that also   $s'$ is interacting, say at  $y_m$ and  $z_k \in \mathcal{N}_{\mathcal{L}^-}(y_m)$.   Consider a straight path $\gamma_1 = (z_1,\ldots,z_k)$ connecting $z_1$ and $z_k$ in $\mathcal{L}^-$, as well as $\gamma_2 = ( x_1,\ldots, x_l)$ and $\gamma_3 = ( y_j,\ldots,y_m)$.  The closed path  (in the sense of graph theory)  $\gamma_1 \circ \gamma_2 \circ \gamma_3$ forms a $(k+l+m-j )$-gon. The sum of its interior angles is $(k+l+m-j - 2 )\pi$. In particular, the interior angles at $z_2,\ldots,z_{k-1}$ are $\pi$, the interior angles at $z_1$ and $z_k$ are in $[\pi/2-\varepsilon,\pi/2+\varepsilon]$, the interior angle at $x_l$ is at least $\pi/2-\varepsilon$, and the  interior angles at $x_1$ and $y_m$ are in $[\pi-\varepsilon,\pi+\varepsilon]$. Thus, as $\theta_{\mathrm{ex}}(\gamma_1) =0$, this implies that $\theta_{\mathrm{ex}}(\gamma_2) + \theta_{\mathrm{ex}}(\gamma_3)  \geq \frac{\pi}{2} - 5\varepsilon  $. This yields a contradiction since  the angle excess of both straight paths is smaller than $ \frac{\pi}{4} -\frac{5}{2}\varepsilon$ by assumption.  

We now prove {\rm (v)}. Assume by contradiction that $\gamma= (x_1,\ldots,x_n)$ is double touching, say in $z_1, z_k \in \mathcal{L}^-$. Consider a straight path $(z_1,\ldots,z_k)$ connecting $z_1$ and $z_k$ in $\mathcal{L}^-$. The path $(x_1,\ldots,x_n,z_k,z_{k-1}\ldots,z_1)$  (in the sense of graph theory)  forms an $(n+k)$-gon. The sum of its interior angles is $(n+k-2)\pi$. The interior angles at $z_2,\ldots,z_{k-1}$ are $\pi$, the interior angles at $z_1$ and $z_k$ are in $[\pi/2-\varepsilon,\pi/2+\varepsilon]$, and the  interior angles at $x_1$ and $x_n$ are in $[\pi-\varepsilon,\pi+\varepsilon]$. This implies that $\theta_{\mathrm{ex}}(\gamma) \geq \pi-4\varepsilon$, a contradiction.

Next, we prove  {\rm (vi)} and refer to Figure~\ref{Fig:Lem3.5(vi)} for an illustration.  Assume by contradiction  that $\hat{s}_1 = \hat{s}_2 =:  \hat{s}$. We denote by $\gamma_1 \subset s_1$ the path connecting  the point $s_1 \cap s'$ and the point $s_1 \cap s_2$. Similarly, we denote $\gamma_2 \subset s_2$  the path connecting  the point $s_1 \cap s_2$ and the point $s'' \cap s_2$, $\gamma' \subset s'$ the path connecting  the point $s' \cap \hat{s}$ and the point $s' \cap s_1$, $\gamma'' \subset s''$ the path connecting  the point $s'' \cap s_2$ and the point $s'' \cap \hat{s}$, and $\hat{\gamma} \subset \hat{s}$ the path connecting  the point $s'' \cap \hat{s}$ and the point $s' \cap \hat{s}$. Note that the path $\gamma'\circ \gamma_1 \circ \gamma_2 \circ \gamma'' \circ \hat{\gamma}$ is a closed   and thus, taking into account the fact that all bond angles belong to $[k\pi/2-\varepsilon,k\pi/2+\varepsilon]$, $k=1,2,3$, we obtain that $\theta_{\mathrm{ex}}(\gamma_1) +\theta_{\mathrm{ex}}(\gamma_2) + \theta_{\mathrm{ex}}(\gamma') + \theta_{\mathrm{ex}}(\gamma'')+ \theta_{\mathrm{ex}}(\hat\gamma) \geq \pi/2-5\varepsilon$. This yields a contradiction since  the angle excess of all five straight paths is smaller than $\frac{\pi}{10}-\varepsilon$ by assumption.  

 Lastly, the proof of  {\rm (vii)} follows like the one of  {\rm (vi)}
 by taking as $\hat{s} \subset \mathcal{L}^-$  to be  the path
 connecting $s'$ and $s''$, see Figure~\ref{Fig:Lem3.5(vi)}, 
 observing additionally that there are two  particles  bonded to the
 substrate where the interior angles of the closed path are in
 $[\pi-\varepsilon,\pi+\varepsilon]$, cf.\ the proof of (v). 
 \end{proof}

\begin{figure}[htp]
\begin{tikzpicture}[scale=.75]

\foreach \j in {0,...,10}{
\draw[ultra thin,gray, dashed](.5*\j,-1.5)--++(0,3);
}

\foreach \j in {0,...,6}{
\draw[ultra thin,gray, dashed](0,-1.5+.5*\j)--++(5,0);
}

\draw[white,thick](4.5,-1.5)--++(.5,0)--++(0,.5);

\draw[ultra thin](0,0)--++(5,0);

\draw[ultra thin](2.5,-1.5)--++(0,3);

\draw[ultra thin](5,-1)--++(0,3)++(1,0)--++(-2,0);

\draw[ultra thin](4.5,-1.5)--++(-5,0)++(0,-.5)--++(0,1.5);

\draw[ultra thin, smooth,dashed](5,-1.5)++(-5.5,0)++(0,-.5)++(0,1.5)--(-.5,.5)--(-.5,1)--(-.6,1.2)--(-.7,1.5)--(-.6,1.8)--(.5,2.5)--(3,2.2)--(3.5,2)--(4,2);

\draw(4,.25) node {$s_1$};

\draw(5.25,.75) node {$s'$};

\draw(1.25,-1.75) node {$s''$};

\draw(2.25,-.75) node {$s_2$};

\draw(5.25,2.25) node {$\hat{s}_1$};

\draw(-0.75,-1) node {$\hat{s}_2$};

\draw(-0.3,2.4) node {$\hat{s}$};

\begin{scope}[shift={(8,0)}]

\draw[ultra thin, gray] (-.5,-2) --(7,-2);
\draw[ultra thin, gray] (-.5,-1.5) --(7,-1.5);
\draw[ultra thin, gray] (-.5,-2.5) --(7,-2.5);

\foreach \j in {-1,...,14}{
\draw[ultra thin,gray](.5*\j,-1.5)--++(0,-1);
}

\foreach \j in {0,...,10}{
\draw[ultra thin,gray, dashed](.5*\j,-1.5)--++(0,3);
}

\foreach \j in {0,...,6}{
\draw[ultra thin,gray, dashed](0,-1.5+.5*\j)--++(5,0);
}

\draw[ultra thin, smooth,dashed](5,1.5)--(5.6,1.6)--(6.2,1.8)--(6.8,1.6)--(7.2,1.4)--(7.3,1.2)--(7.3,1)--(7,.8)--(6.8,-.2)--(6.7,-.7)--(6.5,-1.5);

%\draw[white,thick](4.5,-1.5)--++(.5,0)--++(0,.5);

\draw[ultra thin](0,0)--++(5,0);

\draw[ultra thin](2.5,-1.5)--++(0,3);

\draw[ultra thin](5,-1.5)--++(0,3);

\draw[ultra thin](5,1.5)--++(-5,0);

\draw(4,.25) node {$s_1$};

\draw(5.25,.75) node {$s'$};

\draw(1.25,1.75) node {$s''$};

\draw(2.25,.75) node {$s_2$};

\draw(5.70,-1.25) node { $\hat{s}$ };

\end{scope}

\end{tikzpicture}
\caption{On the left: A configuration that satisfies the contradictory assumption used in the proof of~\Cref{lemma:excess}(vi).  On the right: A configuration that satisfies the contradictory assumption used in the proof of~\Cref{lemma:excess}(vii).}
\label{Fig:Lem3.5(vi)}
\end{figure}
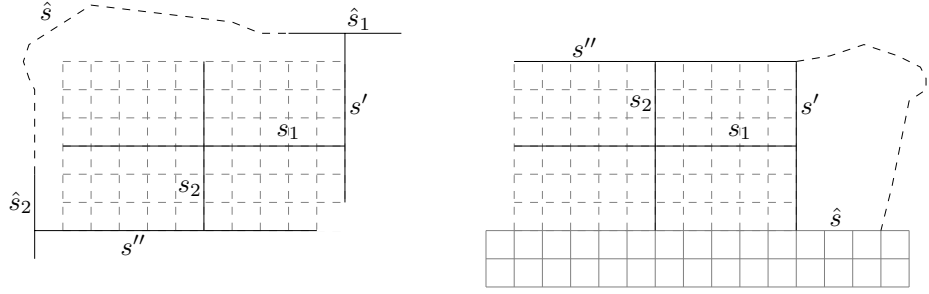
%\end{proof}  

\begin{remark}\label{rem: eps}
In the following, we assume that $\varepsilon_0>0$ is small enough such that $ \pi/10 - 7\varepsilon/5  \leq  \pi/6-  4\varepsilon$ for all $0  < \eps < \eps_0$.  
\end{remark}

 We state a consequence of \Cref{lemma:excess}. (Actually, we only use properties (i), (iv), and (v).)

\begin{lemma}[Graphs  with  open paths, no double touching]\label{lemma:opengraph}  Let $G=(V,E)$ be an $\eps$-regular graph such that $ \max_{\gamma \in \Gamma} \theta_{\mathrm{ex}}(\gamma) < \frac{\pi}{10}-  \frac{7}{5} \eps$.  Then,  
\begin{itemize}
\item[\rm (i)]  $\sum_{s \in \mathcal{S}} l(s) =2N$;
\item[\rm (ii)] $ F_{\rm bond, \beta} (G)  = 2(1-\beta) \#\mathcal{S}_{\rm int} + 2\# \mathcal{S}_{\rm no}$;
\item[\rm (iii)] Each $x\in V$ is contained in exactly two strata $s(x)$ and $s^\perp(x) \in \mathcal{S}^\perp(s(x))$. 
\end{itemize}
\end{lemma}
\begin{proof} Statement {\rm (i)} follows exactly as in~\cite[Proof of Lemma 3.3]{FriedrichKreutz}, so we only need to prove {\rm (ii)} and {\rm (iii)}. 

We first prove {\rm (ii)}. To this end, recalling \eqref{def:Fbond} it suffices to prove the claims  
\begin{align*}
{\rm (a)} \quad  \#\mathcal{S} = \frac{1}{2}\sum_{x \in V} (4-\#\mathcal{N}(x,E))  \quad \text{and} \quad {\rm (b)} \quad\#\mathcal{S}_{\rm int} = \sum_{x\in V}\#\mathcal{N}_{\mathcal{L}^-}(x)\,.
\end{align*}
The proof of {\rm (a)}  can be found in~\cite[Proof of Lemma 3.3]{FriedrichKreutz}. We  prove {\rm (b)} by induction over $ m   \defas  \#\{x \in V \colon \mathcal{N}_{\mathcal{L}^-}(x) \neq \emptyset\}$. The statement is clearly true for $m=0$. Let now $ m   \geq 1$ and fix $x \in V$ such that  $\mathcal{N}_{\mathcal{L}^-}(x) \neq \emptyset$. We have that $\#\mathcal{N}_{\mathcal{L}^-}(x)=1$ by~\Cref{lemma:elementaryprop}. Considering $\hat{G}= (V\setminus \{x\},E \setminus \{\{x,y\} \colon y \in \mathcal{N}(x,E)\})=(\hat{V},\hat{E})$   we get    $\#\{ z  \in \hat{V} \colon \mathcal{N}_{\mathcal{L}^-}(z) \neq \emptyset\} = m-1$. We  can apply the induction hypothesis  to obtain   
\begin{align*}
 \#\mathcal{S}_{\rm int}(\hat{G}) = \sum_{z \in \hat{V}} \#\mathcal{N}_{\mathcal{L}^-}(z) = \sum_{z \in V} \#\mathcal{N}_{\mathcal{L}^-}(z) -1\,.
\end{align*}
It thus remains to show that $ \#\mathcal{S}_{\rm int}(\hat{G}) =  \#\mathcal{S}_{\rm int}(G)-1$. By definition we clearly have $\mathcal{S}_{\rm int}(\hat{G}) \subset \mathcal{S}_{\rm int}({G})$. Consider the unique   $s \in \mathcal{S}_{\rm int}(G)$ with $x \in s$ (uniqueness follows from~\Cref{lemma:excess}(iv)).    It suffices to show that $s \setminus \lbrace x \rbrace \notin  \mathcal{S}_{\rm int}(\hat{G})$.    If $s=(x)$, this is clear.   On the other hand, if $s=(x_1,\ldots,x_n)$, for $n \ge 2$  with $x_1=x$, then, as $s$ is not double touching, we have that $  (x_2,\ldots,x_n) \notin \mathcal{S}_{\rm int}(\hat{G})$. Therefore $ \#\mathcal{S}_{\rm int}(\hat{G}) =  \#\mathcal{S}_{\rm int}(G)-1$.

We now turn to the proof of {\rm (iii)}. Note that, if $x \in V_3\cup V_4 \cup (V_2 \setminus V_2^\pi)$, then there exist two straight paths passing through $x$ and thus there exist $s_1,s_2 \in \mathcal{S}_\Gamma$ that contain $x$. (Note that all $\gamma$ are open so the same stratum cannot pass through $x$ twice). If $x\in V_2^\pi \cup V_1$, there exists a straight path containing $x$ and thus there exists $s \in \mathcal{S}_\Gamma$ containing $x$. By construction, we added the degenerate stratum $(x)$ containing $x$. Therefore, also in this case there are two strata $s_1,s_2 \in \mathcal{S}$ containing $x$. Lastly, for  $x \in V_0$ we added twice the degenerate stratum $(x)$  which contains $x$. This concludes the proof.  
\end{proof}

 We now come to the \textit{stratification} of bond graphs. The following lemma   allows to reduce the problem of crystallization to a purely geometric problem of minimizing the number of strata in graphs containing only open strata  with small angle excess.  Recall \eqref{def:F}--\eqref{def:Fbond}.

\begin{lemma}(Construction of a graph with small angle excess)\label{lemma:construction} Let $ G=(V,E)$ be $\varepsilon$-regular. Then, there exists $G_\mathrm{o}= (V,E_\mathrm{o})$ with $ E_\mathrm{o} \subset E$ such that
\begin{itemize}
\item[\rm (i)] $\max_{\gamma \in \Gamma(G_\mathrm{o})} \theta_{\mathrm{ex}}(\gamma) < \frac{\pi}{10}-7\varepsilon/5$\,;
\item[\rm (ii)]  $G_{\mathrm{o}}$ satisfies
\begin{align*}
 F_\beta (G) \geq   F_{\mathrm{bond},\beta} (  G_{\mathrm{o}})   \,
\end{align*}
with equality only if $E = E_\mathrm{o}$,  $|x-y|=1$ for all $x \in V$, $y\in \mathcal{N}(x,E)\cup \mathcal{N}_{\mathcal{L}^-}(x)$,  and $\theta \in \{\pi/2,\pi,3\pi/2\}$ for all $\theta=\theta_{x,y,z}$ with $x,y,z \in V \cup \mathcal{L}^-$ such that $|x-y|=|z-y|=1$.
\end{itemize}
\end{lemma} 
\begin{proof}
The  proof can be found in~\cite[Proof of Lemma 3.7]{FriedrichKreutz}, replacing $\frac{\pi}{6}$ by $\frac{\pi}{10}$. (Compare the assumption in ($\rm iii_3$)  with the one in~\cite{FriedrichKreutz}.)
\end{proof}

 We now come to an estimate on $ F_{\mathrm{bond},\beta}$ for graphs satisfying \Cref{lemma:construction}(i). This will be the key ingredient for the proof of \Cref{thm:crystallization} in Section~\ref{sec: main}.  Recall Definition \ref{def: strata-ort}.

\begin{lemma}[Estimate on $F_{\mathrm{bond},\beta}$, $\beta \in (0,1)$]\label{lem:orthogonalstrata} Let $0 <\eps < \eps_0  $ and $\beta \in (0,1)$. Let  $G=(V,E)$ be an $\eps$-regular graph such that
\begin{align}\label{eq: smalli}
 \max_{\gamma \in \Gamma} \theta_{\mathrm{ex}}(\gamma) < \frac{\pi}{10}-  \frac{7}{5} \eps  \,.
\end{align}
Let $s_1 \in \mathcal{S}$ and $s_2 \in \mathcal{S}^\perp(s_1)$.  Then, the following holds true:
\begin{itemize}
\item[{\rm (i)}] $ F_{\mathrm{bond}, \beta }  (G) \geq 2(1-\beta) \max_{s \in s_1,s_2} l(s) + 2\min_{s \in s_1,s_2} l(s)$\,.
\item[{\rm (ii)}] If 
\begin{align}\label{incl:nospan}
\mathrm{span}(s_1)\subsetneq V  \quad \text{and}\quad \mathrm{span}(s_2)\subsetneq  V   \,,
\end{align}
then 
\begin{align*}
 F_{\mathrm{bond},\beta}  (G) \geq 2(1-\beta) \max_{s \in s_1,s_2} l(s) + 2\min_{s \in s_1,s_2} l(s) + (4- 2  \beta) \,.
\end{align*} 
\end{itemize}
\end{lemma}
\begin{proof} 
By  \eqref{eq: smalli} and \Cref{rem: eps} we observe that all properties in \Cref{lemma:excess} hold. Therefore, also the properties in \Cref{lemma:opengraph} are satisfied.

\noindent\emph{Proof of {\rm (i)}:} First of all, by \Cref{lemma:excess}(i) all $s\in \Gamma$ (and thus all $s \in \mathcal{S})$ are open.  By~\Cref{lemma:excess}(iii) we have $\mathcal{S}^\perp(s_1) \cap \mathcal{S}^\perp(s_2)=\emptyset$.  Moreover, note that, if there is $s_0 \in  \mathcal{S}^\perp(s_1) $ (or $s_0 \in  \mathcal{S}^\perp(s_2) $, respectively)  such that $s_0 \in \mathcal{S}_{\rm int}$, then,  due to \Cref{lemma:excess}(vii), all $\hat{s} \in \mathcal{S}^\perp(s_2)$ (or $\hat{s} \in \mathcal{S}^\perp(s_1)$, respectively)   satisfy $\hat{s} \in \mathcal{S}_{\rm no}$.   Together with  \Cref{lemma:opengraph}(ii) and   \Cref{lemma:excess}(ii), this implies
\begin{align*}
 F_{\mathrm{bond},\beta} (G)  =  2(1-\beta) \#\mathcal{S}_{\rm int} + 2\#\mathcal{S}_{\rm no} \geq 2(1-\beta)\max_{s \in s_1,s_2} l(s) + 2\min_{s \in s_1,s_2} l(s) \,.
\end{align*}
This concludes the proof of (i). \\ 
\noindent\emph{Proof of {\rm (ii)}:} If \eqref{incl:nospan} holds, then we can find $x \in V \setminus \mathrm{span}(s_1)$. There are two cases to consider:
\begin{itemize}
\item[(a)] $x \notin \mathrm{span}(s_2)$;
\item[(b)] $x \in \mathrm{span}(s_2)$.
\end{itemize}
\noindent \emph{Case {\rm (a)}:} In this case, there are two strata $\hat{s}_1,\hat{s}_2 \in \mathcal{S} \setminus\left( \mathcal{S}^\perp(s_1) \cup \mathcal{S}^\perp(s_2)\right)$ such that $x \in  \hat{s}_1\cup\hat{s}_2$ and thus $\hat{s}_1 \in \mathcal{S}^\perp(\hat{s}_2)$. By \Cref{lemma:excess}(iv) at most one of the two strata $\hat{s}_1, \hat{s}_2$ belongs to $\mathcal{S}_{\rm int}$. Now, as in the proof of (i), we obtain 
\begin{align*}
 F_{\mathrm{bond},\beta}  (G)  =  2(1-\beta) \#\mathcal{S}_{\rm int} + 2\#\mathcal{S}_{\rm no} \geq  2(1-\beta) \max_{s \in s_1,s_2} (l(s)+1) + 2\min_{s \in s_1,s_2} (l(s)+1) \,.
\end{align*}
This concludes (ii) in case (a). \\ 
\noindent \emph{Case {\rm (b)}:}  As $\mathrm{span}(s_2) \subsetneq V$, there exists $z \in V\setminus \mathrm{span}(s_2)$. We can assume that $z \in \mathrm{span}(s_1)$  as otherwise the estimate follows by repeating the argument in (a).  Now, due to \Cref{lemma:excess}(iii), for $s_1^\perp,s_2^\perp\in \mathcal{S}^\perp(s_2)$ we have that $s_1^\perp \cap s_2^\perp=\emptyset$,  and thus   there exists  exactly  one stratum $s''\in S^\perp(s_2)$ such that $s''\cap \{x\} \neq \emptyset$. In the same fashion, there exists exactly one  $s' \in S^\perp(s_1)$  such that $s'\cap \{z\} \neq \emptyset$.  We also note that $s\cap \{x\} = \emptyset$ for all $s\in \mathcal{S}^\perp(s_1)$ and $s\cap \{z\} = \emptyset$ for all $s\in \mathcal{S}^\perp(s_2)$ since $x \in V \setminus \mathrm{span}(s_1)$ and $z \in V\setminus \mathrm{span}(s_2)$. By~\Cref{lemma:opengraph}(iii),  for  $x\in V$ there exist two strata $\hat{s}_2, s''$ such that $x\in \hat{s}_2, s''$ and therefore (by definition) $\hat{s}_2 \in  S^\perp(s'')  $. Similarly, there exist two strata $\hat{s}_1, s'$ such that $ z \in \hat{s}_1, s''$ and therefore $\hat{s}_1 \in S^\perp(s')$.  We get 
\begin{align}\label{5X}
\hat{s}_1, \hat{s}_2 \notin \mathcal{S}^\perp(s_1), \quad \quad  \hat{s}_1, \hat{s}_2 \notin \mathcal{S}^\perp(s_2)\,.
\end{align}
Indeed,  as $x \in s''$ with $s'' \in \mathcal{S}^\perp(s_2)$, by~\Cref{lemma:excess}(iii) we have that  $\hat{s}_2 \notin \mathcal{S}^\perp(s_2)$ and, by the same reasoning, $\hat{s}_1 \notin \mathcal{S}^\perp(s_1)$. As $x \notin \mathrm{span}(s_1)$, we clearly have that $\hat{s}_2 \notin \mathcal{S}^\perp(s_1)$ and analogously $\hat{s}_1 \notin \mathcal{S}^\perp(s_2)$. 

Moreover,  by~\Cref{lemma:excess}(vi) we have that $\hat{s}_1 \neq \hat{s}_2$, and, if both strata $\hat{s}_1, \hat{s}_2$ belong to $\mathcal{S}_{\rm int}$,  we necessarily have    $s', s''\in \mathcal{S}_{\rm no}$ by \Cref{lemma:excess}(iv).  Arguing as in the proof of (i) and taking \eqref{5X} into account, we  therefore obtain 
\begin{align*}
 F_{\mathrm{bond},\beta} (G)  =  2(1-\beta) \#\mathcal{S}_{\rm int} + 2\#\mathcal{S}_{\rm no} \geq 2(1-\beta)\max_{s \in s_1,s_2} (l(s)+1) + 2\min_{s \in s_1,s_2} (l(s)+1) \,.
\end{align*}
This concludes the proof. 
\end{proof}

\begin{lemma}[Estimate on $F_{\mathrm{bond},\beta}$, $\beta \ge 1$]\label{lem:orthogonalstrata2}  Let $0 <\eps < \eps_0  $ and $\beta \ge 1$. Let  $G=(V,E)$ be an $\eps$-regular graph satisfying \eqref{eq: smalli}. Then,  $ F_{\mathrm{bond}, \beta }  (G) \geq 2(1-\beta) N + 2$ with equality if and only if $\#\mathcal{S}_{\rm int} = N$.  
\end{lemma}

\begin{proof}
 \Cref{lemma:opengraph}(iii) and \Cref{lemma:excess}(iv) clearly imply $\#\mathcal{S}_{\rm int} \le N$ and $\#\mathcal{S}_{\rm no} \ge 1$.  This along with the identity  $F_{\mathrm{bond},\beta} (G)  =  2(1-\beta) \#\mathcal{S}_{\rm int} + 2\#\mathcal{S}_{\rm no}$  from  \Cref{lemma:opengraph}(ii) yields the statement. 
\end{proof}

\section{Proof of the main results}\label{sec: main}

 This section is devoted to the proofs of Theorem  \Cref{thm:crystallization} and \Cref{thm:fluctuation}.

\subsection{Crystallization}  \label{sec:thm:crystallization}
 For the proof of Theorem  \Cref{thm:crystallization}, our strategy is to  show that the minimum of $F_\beta$ is given by $m_\beta(N)$, and that it is attained by subsets of $\Z^2$ that touch the substrate. In view of \eqref{eq: relation}, this will show the result.  Recall the definition of $G_{\mathrm{nat}}$ in \eqref{eq: relation-new}.  We first state the following upper bound.

\begin{lemma}(Upper bound) \label{lemma:energybound}Let $C_N$ be a minimizer of \eqref{eq: main energy}. Then, $G_{\mathrm{nat}}$ satisfies
\begin{align*}
 F_\beta  (G_{\mathrm{nat}}) \leq m_\beta(N)\,.
\end{align*}
\end{lemma}
\begin{proof} Fix $N \in \mathbb{N}$ and let $h \in \{1,\ldots,N\}$ be such that $2h+ 2(1-\beta) \big\lceil \frac{N}{h}\big\rceil = m_\beta(N)$. We have two cases:
\begin{itemize}
\item[(i)] $\frac{N}{h} \notin \mathbb{N}$. In this case, $N=h\lfloor \frac{N}{h}\rfloor +k$ for some $k \in \{1,\ldots,h-1\}$, and we define 
\begin{align*}
 \bar{C}_N  =\left(\left\{1,\ldots,\left\lceil \frac{N}{h}\right\rceil-1\right\} \times \{1,\ldots,h\}\right) \cup \left( \left\{\left\lceil \frac{N}{h}\right\rceil \right\}\times\{1,\ldots,k\}\right)\,.
\end{align*}
\item[(ii)] $\frac{N}{h} \in \mathbb{N}$. Here, one chooses $ \bar{C}_N  = \{1,\ldots\frac{N}{h}\} \times \{1,\ldots,h\}$.
\end{itemize}
%
%
%Note that $h \cdot\big\lceil \frac{N}{h}\big\rceil \geq N$ and thus there exists $k \in \{1,\ldots,h \}$ such that $N = h \cdot(\big\lceil \frac{N}{h}\big\rceil-1) +k$. Now define
%\begin{align*}
%C_N =\left(\left\{1,\ldots,\left\lceil \frac{N}{h}\right\rceil-1\right\} \times \{1,\ldots,h\}\right) \cup \left( \left\{\left\lceil \frac{N}{h}\right\rceil \right\}\times\{1,\ldots,k\}\right)\,.
%\end{align*}
 In view of \eqref{def:Fbond},  by directly checking that $F_\beta(\bar{G}_{\rm nat} ) =F_{\mathrm{bond},\beta}(\bar{G}_{\rm nat}) = 2h+ 2(1-\beta) \big\lceil \frac{N}{h}\big\rceil = m_\beta(N)$, where $ \bar{G}_{\rm nat} $ denotes the natural bond graph related to $\bar{C}_N$, one concludes the proof.
\end{proof}

 Given the upper bound of \Cref{lemma:energybound}, the proof of
\Cref{thm:crystallization} follows by checking the corresponding lower
bound. 

 %The core of the proof now consists in proving a lower bound. 

\begin{proof}[Proof of \Cref{thm:crystallization}] Let $C_N$ be a minimizer of \eqref{eq: main energy}.  Then, its natural bond graph $G_{\mathrm{nat}}$ is $\eps$-regular by \Cref{lemma:elementaryprop}.  We denote by $G_\mathrm{o}=(V,E_\mathrm{o})$ the graph obtained in \Cref{lemma:construction}. The graph $G_\mathrm{o}$ is also $\varepsilon$-regular and satisfies  
\begin{align*}
%\label{ineq:prop}
\max_{\gamma \in \Gamma(G_o)}\theta_{\mathrm{ex}}(\gamma) < \pi/10- 7\varepsilon/5\,,  
\end{align*} 
 i.e., all properties derived in Section \ref{subsec:stratbondgraph} hold for the graph $G_\mathrm{o}$. 
 The main part of the proof consists in verifying  
\begin{align}\label{eq: main claim}
F_{\mathrm{bond},\beta}(  G_{\mathrm{o}}  ) \geq m_\beta(N)\,.
\end{align}
 Once  \eqref{eq: main claim} is proven,  we conclude as
 follows. First, \eqref{eq: main claim0} holds due to
 \Cref{lemma:construction}(ii)  and \eqref{eq: relation}. To
 characterize the equality case, we get from
 \Cref{lemma:construction} that  $G=G_\mathrm{o}$, that all bond
 lengths are $1$ (including the ones between $V$ and $\mathcal{L}^-$),
 and all bond angles belong to $\{\pi/2,\pi,3/2\pi\}$. This shows that
 each connected component (in the sense of graph theory) of $G$ lies
 in a rotated and shifted version of $\Z^2$. If there was more than
 one connected component, one could obtain a modified configuration
 with an additional bond. This contradicts minimality. Furthermore, if
 $\mathcal{S}_{\rm int} =\emptyset$, one could obtain a modified
 configuration with strictly less energy than the original one by
 placing a rotated and shifted version of $G$ in such a way  that $V
 \subset \Z^2$ and $ \mathcal{N}_{\mathcal{L}^-}(x) \neq \emptyset$
 for some $x \in V$.  As a consequence, we get that  $ V\subset \Z^2
  \cap  \lbrace x_2>0 \rbrace $  with length   $l= \#\{k\colon
 (\{k\} \times \mathbb{Z}) \cap V \neq \emptyset\}$ and height $h
 =\#\{k\colon (\mathbb{Z}\times \{k\}) \cap  V \neq \emptyset\}$, where
 the properties stated in  \Cref{rem:convexity etc.}(i),(ii) 
  follow    by repeating the  argument  in~\cite[Proposition~6.3]{Mainini-Piovano}.  This also shows \Cref{rem:convexity etc.}(iii)  and eventually the remaining property   
  $\#(V \cap \{x_2=1\}) \geq   N/ h_*(\beta,N) $.

We now show  \eqref{eq: main claim}.  In the following, for simplicity  we write $\mathcal{S}$, $\mathcal{S}_{\rm int}$,  and   $\mathcal{S}_{\rm no}$ in place of $\mathcal{S}(G_{\mathrm{o}})$, $\mathcal{S}_{\rm int}(G_{\mathrm{o}})$, and  $\mathcal{S}_{\rm no}(G_{\mathrm{o}})$. The case $\beta \ge 1$ directly follows from \Cref{lem:orthogonalstrata2}, so we focus on the case $\beta \in (0,1)$.   Recalling \eqref{def:span} we distinguish two cases:
\begin{itemize}
\item[(i)] There exists $s\in \mathcal{S}$ such that
\begin{align}\label{eq:spanning}
 {\rm span}(s)    =V\,;
\end{align} 
\item[(ii)] For all $s \in \mathcal{S}$  we have
\begin{align}\label{incl:strict}
 {\rm span}(s)   \subsetneq V\,. 
\end{align} 
\end{itemize}

\noindent \emph{Case (i):} Select $s\in \mathcal{S}$ such that  \eqref{eq:spanning} holds true. By \eqref{eq:spanning} and \Cref{lemma:excess}(ii),(iii), choosing some $\hat{s} \in {\rm argmax}_{s' \in \mathcal{S}^\perp(s)}\, l(s')  $ we get $l(s)\cdot l(\hat{s}) \ge \sum_{s' \in \mathcal{S}^\perp(s)}\, l(s')    = N$.  Assume without loss of generality that $l(\hat{s})  \geq l(s)$. Then, $l(\hat{s}) \geq \lceil N/l(s)\rceil$. Using \Cref{lem:orthogonalstrata}(i), we obtain
\begin{align*}
F_{\mathrm{bond},\beta}(  G_{\mathrm{o}}) \geq 2l(s) + 2(1-\beta)  l(\hat{s}) \geq  2 l(s)  + 2(1-\beta) \lceil N/l(s)\rceil \geq m_\beta(N)\,.
\end{align*}
This concludes Case (i).

\noindent \emph{Case (ii):}  Our goal is to find $N_\beta$ such that for all $N \ge N_\beta$ the estimate
\begin{align}\label{ineq:cardstrata}
2\#\mathcal{S}_{\rm no} +  2(1-\beta)\#\mathcal{S}_{\rm int} = F_{\mathrm{bond},\beta}( G_{\mathrm{o}} ) > m_\beta(N)
\end{align}
 holds which shows that this case can never occur for a minimizer. Define $p \colon \mathcal{S} \to \lbrace \frac{1}{2} ,  \frac{1}{2(1-\beta)} \rbrace  $ by
\begin{align}\label{def:p}
p(s) = \begin{cases} \frac{1}{2} &\text{if } s \in \mathcal{S}_{\rm no}\,,\\
\frac{1}{2(1-\beta)} &\text{if } s \in \mathcal{S}_{\rm int}\,.\\
\end{cases}
\end{align}
We prove that there exists $x_0 \in V$ such that
\begin{align}\label{ineq:biggestcross}
p(s(x_0)) l(s(x_0)) + p(s^\perp(x_0)) l(s^\perp(x_0)) \geq F_{\mathrm{bond},\beta}(G_{\mathrm{o}})^{-1} 4N\,,
\end{align}
where $s(x_0),s^\perp(x_0)\in \mathcal{S}$ are the two strata such that $ x_0 \in   s(x_0)\cap s^\perp(x_0)$, see \Cref{lemma:opengraph}(iii).  To this end, define $\mu(s)  = \frac{1}{p(s)}$. By \Cref{lemma:opengraph}(ii) we have 
\begin{align}\label{eq:masmu}
\mu(\mathcal{S}):= \sum_{s \in \mathcal{S}} \mu(s) = 2\#\mathcal{S}_{\rm no} + 2(1-\beta) \#\mathcal{S}_{\rm int}=F_{\mathrm{bond},\beta}(G_{\mathrm{o}})\,.
\end{align}
Now,  by Jensen's inequality together with \Cref{lemma:opengraph}(i) and \eqref{eq:masmu} we obtain
\begin{align}\label{ineq:jensen1}
\sum_{s\in\mathcal{S}} \left(p(s)l(s)\right)^2\mu(s)  &\geq \mu(\mathcal{S})^{-1} \left(\sum_{s\in\mathcal{S}} p(s)l(s)\mu(s)\right)^2 = \mu(\mathcal{S})^{-1} \left(\sum_{s\in \mathcal{S}}l(s)\right)^2 \notag \\&= \mu(\mathcal{S})^{-1} 4N^2 = F_{\mathrm{bond},\beta}(G_{\mathrm{o}})^{-1} 4 N^2\,.
\end{align}
On the other hand, letting $s(x),s^\perp(x)\in \mathcal{S}$ be the two strata such that $ x \in   s(x)\cap s^\perp(x)$, exchanging the order of summation, and using \Cref{lemma:opengraph}(iii), we obtain 
\begin{align}\label{eq:crosssum}\begin{split}
\sum_{s\in\mathcal{S}} \left(p(s)l(s)\right)^2\mu(s)  &= \sum_{s\in \mathcal{S}} p(s)l(s)^2 = \sum_{s \in \mathcal{S}} \sum_{x \in s} p(s)l(s) = \sum_{s \in \mathcal{S}} \sum_{x \in V} p(s)l(s) \mathrm{1}_s(x) \\&=   \sum_{x \in V} \sum_{s \in \mathcal{S}} p(s)l(s) \mathrm{1}_s(x) = \sum_{x\in V} \big(p(s(x))l(s(x)) + p(s^\perp(x))l(s^\perp(x)) \big)\,.
\end{split}
\end{align}
 Inequality  \eqref{ineq:biggestcross} follows from \eqref{ineq:jensen1} and \eqref{eq:crosssum} by selecting $x_0 \in V$ such that $p(s(x_0))l(s(x_0)) + p(s^\perp(x_0))l(s^\perp(x_0)) \geq p(s(x))l(s(x)) + p(s^\perp(x))l(s^\perp(x))$ for all $x \in V$.
 Next,  using \eqref{incl:strict} and~\Cref{lem:orthogonalstrata}(ii)  for $s(x_0)$ and $s^\perp(x_0)$,  we obtain   
\begin{align*}
F_{\mathrm{bond},\beta}(G_{\mathrm{o}}) \geq 2(1-\beta) \max_{s \in s(x_0),s^\perp(x_0)}  l(s) + 2 \min_{s \in s(x_0),s^\perp(x_0)}  l(s)  +  (4- 2 \beta)\,.
\end{align*}
Now, by \eqref{def:p} and by multiplying  \eqref{ineq:biggestcross} with $4(1-\beta) $ we obtain  
\begin{align*}
F_{\mathrm{bond},\beta}(G_{\mathrm{o}}) \geq F_{\mathrm{bond},\beta}(G_{\mathrm{o}})^{-1} 16(1-\beta)  N + (4- 2\beta)\,.
\end{align*}
By solving the quadratic equation $t^2-(4-2\beta)t -16(1-\beta)  N \geq 0$ for $t \geq 0$, we obtain 
\begin{align}\label{ineq:lowerboundfbond}
F_{\mathrm{bond},\beta}(G_{\mathrm{o}}) \geq \sqrt{(2-\beta)^2 +16(1-\beta) N} + 2-\beta\,.
\end{align}
Let $ \overline{h} = \lceil ((1-\beta) N)^{1/2}\rceil$, $h_c =
((1-\beta) N)^{1/2}$, and $f_N(h) = 2h + 2(1-\beta) \frac{N}{h}$. We
then have $|h_c-\overline{h}|\leq 1$ and using that $\min_h f_N(h) =
f_N(h_c)$, by Taylor   expansion,  together with the fact that $\overline{h} \geq h_c$ and thus $f''_N(h_c)\geq f''_N(h)$ for all $h \in [h_c,\overline{h}]$, we get
\begin{align*}
2\overline{h}+ 2(1-\beta) \frac{N}{\overline{h}}= f_N(\overline{h}) \leq  f_N(h_c) + \frac{1}{2}f_N''(h_c)(\overline{h}-h_c)^2 =4\sqrt{(1-\beta)  N} +  2  (1-\beta)^{-\frac{1}{2}}N^{-\frac{1}{2}}\,.
\end{align*}
Therefore, using \eqref{ineq:lowerboundfbond}, we obtain 
\begin{align*}
m_\beta(N) &\leq 2\overline{h} + 2(1-\beta) \left\lceil \frac{N}{\overline{h}} \right\rceil \leq  2\overline{h} + 2(1-\beta)\left( \frac{N}{\overline{h}} +1\right)\\&\leq  2(1-\beta) +4\sqrt{(1-\beta) N} + 2 (1-\beta)^{-\frac{1}{2}}N^{-\frac{1}{2}} \\& <  - \beta + 2 (1-\beta)^{-\frac{1}{2}}N^{-\frac{1}{2}}  +F_{\mathrm{bond},\beta}(G_{\mathrm{o}}) \,.
\end{align*}
 For $N \geq N_\beta  \defas   \beta^{-2} 4  (1-\beta)^{-1}$ this shows \eqref{ineq:cardstrata} and concludes the proof.
\end{proof}

%Estimate \eqref{eq: slice} is related to proving an isoperimetric inequality with respect to the $l^1$-perimeter via slicing. We present a corresponding argument in the continuum setting in Appendix~\ref{appendix2}.  

\subsection{Fluctuation estimates}\label{sec:fluctuation}

The goal of this section is to prove \Cref{thm:fluctuation}. 
Notice that we assume  $\beta \in (0,1)$ throughout this section. 

\begin{definition}\label{def:energies1d2d}
Given $N \in \mathbb{N}$ we define $\mathcal{H}_{\beta,N} \colon \{1,\ldots,N\} \to \mathbb{R}$   and $\mathcal{L}_{\beta,N} \colon \{1,\ldots,N\} \to \mathbb{R}$ by
\begin{align}\label{def:HL}
\mathcal{H}_{\beta,N}(h) = 2h + 2(1-\beta) \left\lceil \frac{N}{h} \right\rceil\,, \quad  \mathcal{L}_{\beta,N}(l) = 2\left\lceil \frac{N}{l} \right\rceil + 2(1-\beta) l\,,
\end{align}
and $\mathcal{E}_\beta \colon \mathbb{N} \times \mathbb{N} \to \mathbb{R}$ by 
\begin{align*}
\mathcal{E}_\beta(h,l) = 2h + 2(1-\beta) l\,.
\end{align*}
Moreover, we set
\begin{align}\label{def:hstar}
h_* =(1-\beta)^{1/2}  N^{1/2} \,, \quad l_* = (1-\beta)^{-1/2}  N^{1/2} \,.
\end{align}
\end{definition}
We observe that 
\begin{align*}
%\label{eq:equivalence}
m_\beta(N) &= \min_h \mathcal{H}_{\beta,N}(h) =\min_l \mathcal{L}_{\beta,N}(l) = \min \left\{\mathcal{E}_\beta(h,l) \colon h,l \in \mathbb{N}, h\cdot l \geq N \right\}\,,
\end{align*}
and 
\begin{align*}
2h_* + 2(1-\beta) \frac{N}{h_*} = \min_{h>0} \Big(2h +  2(1-\beta) \frac{N}{h}\Big)\,, \ \ \quad    2\frac{N}{l_*} +  2(1-\beta)  l_*= \min_{l>0} \Big(2\frac{N}{l} + 2(1-\beta)  \Big)\,.
\end{align*}
Thus, as the next lemma shows, $h_*$ and $l_*$ provide good reference values for minimizers of $\mathcal{H}_{\beta,N}$  and $\mathcal{L}_{\beta,N}$, respectively.

\begin{lemma} \label{lem:quadratic}  
There exists a constant $C_\beta>0$ such that for all minimizers  $h \in \mathbb{N}$ and  $l \in \mathbb{N}$  of $\mathcal{H}_{\beta,N}$ and  $\mathcal{L}_{\beta,N}$, respectively,  we have
\begin{align}\label{ineq:fluctuations}
|h-h_*| \leq C_\beta N^{1/4}\,, \quad \quad  |l-l_*|\leq C_\beta N^{1/4} \,,
\end{align}
where $h_*$ and $l_*$ are given by \eqref{def:hstar}. Furthermore, if $\beta \in \mathbb{R} \setminus \mathbb{Q}$, then the minimizers of $\mathcal{H}_{\beta,N}$ and $\mathcal{L}_{\beta,N}$ are unique.
\end{lemma}

\begin{proof} 
Note that the proofs   for minimizers of  $\mathcal{H}_{\beta,N}$ and  $\mathcal{L}_{\beta,N}$ are analogous. We therefore only prove the statements for minimizers of  $\mathcal{H}_{\beta,N}$.  We first show that for $\beta \in \mathbb{R} \setminus \mathbb{Q}$ the minimizer of $\mathcal{H}_{\beta,N}$ is unique. To this end, assume by contradiction that there exist two minimizers $h_1,h_2 \in \{1,\ldots,N\}$. Then,
\begin{align*}
2h_1 +  2(1-\beta)\left\lceil \frac{N}{h_1} \right\rceil = m_\beta(N) = 2h_2 + 2(1-\beta) \left\lceil \frac{N}{h_2} \right\rceil\,,
\end{align*}
which is equivalent to $2(h_1-h_2) =  2(1-\beta) (\lceil \frac{N}{h_2}
\rceil -  \lceil \frac{N}{h_1}  \rceil)$.  However, this last
equation cannot hold  as $ 2(1-\beta)  \in \mathbb{R}\setminus \mathbb{Q}$.

Next, we prove \eqref{ineq:fluctuations}. Let $h \in \N$ be a minimizer of $\mathcal{H}_{\beta,N}$ and define $h^* = \lceil h_*\rceil$. Then
\begin{align*}
2h + 2(1-\beta) \Big\lceil \frac{N}{h} \Big\rceil  &=\mathcal{H}_{\beta,N}(h) \leq \mathcal{H}_{\beta,N}(h^*) = 2 \lceil h_*\rceil + 2(1-\beta) \left\lceil \frac{N}{\lceil h_*\rceil} \right\rceil \\&\leq 2(h_*+1) + 2(1-\beta) \left(\frac{N}{ h_*}+1\right) \\&= 4-2\beta +\min\left\{2h + 2(1-\beta)\frac{N}{h} \colon h>0 \right\} =4-2\beta + 4\left((1-\beta)  N\right)^{\frac{1}{2}}\,.
\end{align*}
Multiplying  by  $h$, we obtain
\begin{align*}
2h^2 -\left(4-2\beta + 4\left((1-\beta)  N\right)^{\frac{1}{2}}\right)h + 2(1-\beta) N \leq 0\,.
\end{align*}
  The left-hand side  is quadratic   in $h$ with
 zeros  at   $h_\pm = \frac{2-\beta}{2} +h_* \pm
 \frac{1}{4}\left((4-2\beta)^2 +16(2-\beta)h_*\right)^\frac{1}{2}$ and
 thus $h_- \leq h \leq h_+$.   Noting that  $h_* \leq N^{\frac{1}{2}}$  by \eqref{def:hstar},  this implies
\begin{align*}
|h-h_*| \leq \left(\frac{2-\beta}{2} +\frac{1}{4}32^\frac{1}{2}(2-\beta)^\frac{1}{2}\right)N^\frac{1}{4}\,.
\end{align*} 
By choosing  $C_\beta := \frac{2-\beta}{2} +\frac{1}{4}32^\frac{1}{2}(2-\beta)^\frac{1}{2} $    we conclude the proof.
\end{proof}

 We now first prove \Cref{thm:fluctuation} for $\beta \in \mathbb{Q}$. 

\begin{proof}[Proof of \Cref{thm:fluctuation},  upper bound  {\rm ($U_1$)}]
  Let  $\beta \in \mathbb{Q}$,  $N\geq N_\beta$,  with $N_\beta$ from
  \Cref{thm:crystallization},  and let $C_N$ be a minimizer of
  $\mathcal{F}_\beta$. By \Cref{thm:crystallization} we have that $C_N
  \subset \mathbb{Z}^2$.  By~\Cref{lem:quadratic} it suffices to
  prove that there exists a horizontal translation in $\mathbb{Z} \times \{0\}$ and   $C_\beta >0$  such that
\begin{align}\label{ineq:mainineqU1U2}
\#\left(C_{N}  \triangle \left(R_{\max}(\beta,N) +\tau\right)\right) \leq  C_\beta N^{3/4} \,,
\end{align}
where $R_{\max}(\beta,N) =\{1,\ldots,l_{\max}\}\times \{1,\ldots,h_{\max}\}$ with   
\begin{align*}
l_{\max}:= \max\{l \colon \mathcal{L}_{\beta,N}(l) =m_\beta(N)\}\,, \quad h_{\max}:= \max\{h \colon \mathcal{H}_{\beta,N}(h) =m_\beta(N)\}\,.
\end{align*}
Without loss of generality we can assume that the term on the
left-hand side of \eqref{ineq:mainineqU1U2} is minimized    for $\tau=0$. By~\Cref{lem:quadratic}, \Cref{rem:convexity etc.}, and~\Cref{def:energies1d2d}, we have that $C_N \subset R_{\max}(\beta,N)$.  As $\# C_N=N$, it suffices to show that
\begin{align*}
%\label{ineq:cardestimateRmax}
\#R_{\max}(\beta,N) \leq N+C_\beta N^{3/4}\,.
\end{align*}
Now,  by~\Cref{lem:quadratic} we have that $| h_{\rm max}  -h_*| \leq C_\beta N^{1/4}$ and  $ | l_{\rm max}  -l_*|\leq C_\beta N^{1/4} $ with $l_*,h_*$ given in \eqref{def:hstar}. As $l_*,h_* \leq C_\beta N^{1/2}$ and $l_*\cdot h_*=N$, this shows {\rm ($U_1$)}. 
\end{proof}

 \begin{proof}[Proof of \Cref{thm:fluctuation},  lower bound   {\rm ($L_1$)}]   Let again $\beta \in \mathbb{Q}$. Let $\bar{p},\bar{q} \in \mathbb{N}$, $\bar{p} <\bar{q}$, be such that $ 1-\beta  =\frac{\bar{p}}{\bar{q}}$. Let $0<\delta <   \min\{\tfrac{1}{2\bar{p}}, \tfrac{1}{2\bar{q}}\}$ with $\delta^{-1} \in \N$. Set $p = \delta^{-2} \bar{p}$ and $q = \delta^{-2} \bar{q}$.   Let $N_k = k^4 pq +(1-\delta)k^2p \in \N$  for each $k \in \N$.  We claim that
\begin{align}\label{eq:mbetaNk}
m_\beta(N_k) = 2k^2p + 2(1-\beta)(k^2q+1)\,.
\end{align}
We postpone the proof of \eqref{eq:mbetaNk} and show first how we can
conclude once  \eqref{eq:mbetaNk}  is proven. Using \eqref{eq:mbetaNk} and $ 1-\beta  =\frac{p}{q}$, we have that
\begin{align*}
C_{N_k}^1:= \left(\{1,\ldots,k^2q\}\times \{1,\ldots,k^2p\}\right) \cup \left(\{k^2q+1\} \times \{1,\ldots,(1-\delta)k^2p\}\right)
\end{align*}
and 
\begin{align*}
C_{N_k}^2:= &\left(\{1,\ldots,k^2q(1+\delta^2\frac{1}{k})\}\times \{1,\ldots,k^2p(1-\delta^2\frac{1}{k})\}\right) \\&\cup \left(\{k^2q(1+\delta^2\frac{1}{k})+1\}\times \{1,\ldots,(1-\delta)k^2p +\delta^4k^2pq\}\right)
\end{align*}
are minimizers for $N_k$. (In particular, it is elementary to check
that indeed both configurations consist of $N_k$  particles.    Here, we also
used that $k\delta \bar{q} <k - \delta$.)     Furthermore, since $  k^3\ge c_\beta N_k^{3/4}$ for some $c_\beta>0$ depending on
$\beta$,  for  all   $\tau \in \mathbb{Z}\times \{0\}$ we have  
\begin{align*}
\#(C_{N_k}^1\triangle (C_{N_k}^2+\tau)) \geq c_\beta   N_k^{3/4}\,.
\end{align*}
This shows that {\rm ($L_1$)}  must be satisfied either for $C_{N_k}^1$ or $C_{N_k}^2$.

It remains to prove \eqref{eq:mbetaNk}. Testing \eqref{def:mbetan} with $h=pk^2$, we obtain $m_\beta(N_k) \leq 2k^2p + 2(1-\beta) (k^2q+1)$. We now need to prove  $m_\beta(N_k) \geq 2k^2p + 2(1-\beta)(k^2q+1)$. To this end, observe first that for $N_k^0 \defas pqk^4$ we  get
\begin{align}\label{mbeta}
m_\beta(N_k^0) =2pk^2 + 2(1-\beta)  qk^2\,.
\end{align}
Indeed,  we have  $h_*=pk^2,l_*=qk^2 \in \mathbb{N}$ (defined in \eqref{def:hstar}  for $N_k^0$ in place of $N$)  and therefore 
\begin{align*}
2pk^2 + 2(1-\beta)  qk^2 = \mathcal{H}_{\beta, N_k^0  }(h_*) &\geq m_\beta(N_k^0) \geq \min_{h>0} \Big( 2h+ 2(1-\beta) \frac{ N_k^0 }{h}\Big) \\&= 2h_* + 2(1-\beta) \frac{ N_k^0 }{h_*} = 2pk^2 + 2(1-\beta)  qk^2\,.
\end{align*}
This shows \eqref{mbeta}.  By the convexity of  $h\mapsto \frac{N_k^0}{h}$ we further have
\begin{align}\label{mbeta2}
\frac{N_k^0}{h} \ge \frac{N_k^0}{h_*} - \frac{N_k^0}{h_*^2} (h-h_*) = \frac{N_k^0}{h_*} - \frac{q}{p}  (h-h_*) = \frac{N_k^0}{h_*} -  \frac{1}{1-\beta}     (h-h_*) \, .
\end{align}
For  any minimizer $h \in \{1,\ldots, N_k\}$ it holds that  $|h- h_*
| \leq C_\beta N_k^{1/4}  $ by~\Cref{lem:quadratic}.  Along
 with \eqref{mbeta}--\eqref{mbeta2}, for $k$ large enough  
this   shows \begin{align*}
2h + 2(1-\beta) \left\lceil \frac{N_k}{h}\right\rceil &= 2 h_*   +2(h- h_*  ) + 2(1-\beta)  \left\lceil \frac{N_k^0}{h}+(1-\delta)\frac{h_*}{h}\right\rceil  \\&
\geq m_\beta(N_k^0) +2(h-h_*) + 2(1-\beta)  \left\lceil  \frac{1}{1-\beta}    (h_*-h)+(1-\delta)\frac{h_*}{h}\right\rceil   \\&\geq  m_\beta(N_k^0) +2(h-h_*) + 2(1-\beta) \left\lceil  \frac{1}{1-\beta}    (h_*-h)+(1-2\delta)\right\rceil\,.
\end{align*}
Here, we used that $\frac{h_*}{h} \ge 1 - \delta$ for $k$ large enough
such that $C_\beta N_k^{1/4}\leq \frac{1}{2} h_* =
\frac{1}{2}(1-\beta)^{1/2}N_k^{1/2}$ and $ 2 C_\beta (1-\beta)^{-1/2}
N_k^{-1/4}\leq \delta $, recalling  \eqref{def:hstar} and
~\Cref{lem:quadratic}. In view of \eqref{mbeta}, in order to
conclude the proof of \eqref{eq:mbetaNk}, it  suffices to show
$2(h-h_*) + 2(1-\beta)  \left\lceil \frac{1}{1-\beta}
  (h_*-h)-2\delta\right\rceil \geq 0$. Recalling  that  $
1-\beta  =\frac{\bar{p}}{\bar{q}}$ and  setting   $\bar{q}(h-h_*) =\bar{p}l+r$ with $l \in \mathbb{Z}$, $r \in \{0,\ldots,\bar{p}-1\}$, we have
\begin{align*}
 \frac{1}{1-\beta}    (h-h_*) + \left\lceil \frac{1}{1-\beta}    (h_*-h)-2\delta\right\rceil = \frac{1}{\bar{p}}(\bar{p}l+r) +  \left\lceil -\frac{1}{\bar{p}}(\bar{p}l+r) -2\delta\right\rceil = \frac{r}{\bar{p}} +\left\lceil -\frac{r}{\bar{p}}-2\delta\right\rceil\geq 0\,,
\end{align*}
because $-\frac{r}{\bar{p}}-2\delta > -\frac{\bar{p}-1}{\bar{p}} -2 \frac{1}{2\bar{p}}=  - 1 $ as  $\delta< \frac{1}{2\bar{p}}$. This shows \eqref{eq:mbetaNk} and concludes the proof of~{\rm ($L_1$)}.
\end{proof}

 We now consider the case that $\beta \in \R \setminus \mathbb{Q}$ is algebraic. Denoting by $N$ the number of  particles  of a configuration and by $h$ its height, and  by $l$ its length, it is not restrictive to reduce  the problem to considering configurations where $(h-1)$-rows consist of $l$  particles  and the upmost row consists of at most $l$  particles,   cf.\ \Cref{rem:convexity etc.}.  Such configurations will be called \emph{$(h,l)$-configurations} in the following.   If also the upmost row consists of $l$  particles,   we call such a configuration an  \emph{$(h,l)$-rectangle}. In the latter case, we have $N = hl$. Given an  $(h,l)$-configuration, the corresponding energy is given by  $\mathcal{H}_{\beta,N}(h)  = 2 h  + 2(1-\beta)   \lceil \frac{N}{h}  \rceil$, as defined in \eqref{def:HL}. We recall  by \Cref{lem:quadratic}  that,  due to the fact that $\beta$ is irrational, the minimizer $h$ of $\mathcal{H}_{\beta,N}$ is unique.   Moreover, the energy of different $(h,l)$-rectangles is necessarily different.

We fix $N_0 \in \N$ sufficiently large. Our strategy consists in characterizing the optimal $(h,l)$-rectangles with   particle  number $N$ with $N \in [N_0,2N_0]$.  By~\Cref{lem:quadratic} we know that  the unique optimal height corresponding to  $N$,  denoted by $h_*(\beta,N)$ (see \eqref{hstern}),  is given by 
\begin{align}\label{hrange}
h_*(\beta,N) =  \sqrt{1-\beta } N^{1/2} + O(N^{1/4}).
\end{align}
 This shows that for  particle  numbers  $N \in [N_0,2N_0]$ it suffices to consider heights between $c_\beta\sqrt{N_0}$ and $C_\beta\sqrt{N_0}$  for some  $0<c_\beta<C_\beta$, provided that $N_0$ is chosen sufficiently large.  (For technical reasons we will also consider a slightly larger interval.)   Given  such a height $h_0 $,  i.e.,  $c_\beta\sqrt{N_0} \le h_0 \le C_\beta\sqrt{N_0}$,   we choose the unique $l(h_0)$ such that 
$$l(h_0) = {\rm argmin}_l \,   \Big| \frac{1}{1-\beta}     - \frac{l}{h_0}\Big|.  $$
It is elementary to check that  
\begin{align}\label{best approx}
\Big| \frac{1}{1-\beta}      - \frac{l(h_0)}{h_0}\Big| \le \frac{1}{h_0}. 
\end{align}
In the sequel, we will consider $(h_0,l)$-rectangles, and we want to assess whether they are optimal or not. Optimality  corresponds to 
$$h_0  = {\rm argmin}_h \, \mathcal{H}_{\beta,h_0l}    (h). $$
The main properties are the following.

\begin{lemma}\label{lemma: counting}
For   $0<\delta<\frac{1}{6}$ there exists $ N_{\beta,\delta} \in \N$ such that for all $N_0 \ge N_{\beta,\delta}  $ and each $ \frac{1}{2}  c_\beta\sqrt{N_0} \le h_0 \le  2  C_\beta\sqrt{N_0}$  the following holds:
\begin{itemize}
\item[(i)] For all $l \in \N$ with  $|l-l(h_0)| \le N_0^{1/6 - \delta}$, the $(h_0,l)$-rectangle is optimal. 
\item[(ii)] For all $l \in \N$ with  $|l-l(h_0)| \ge N_0^{1/6 + \delta}$, the $(h_0,l)$-rectangle is not optimal.   
\end{itemize}
\end{lemma}

This shows that among configurations with  particle  number between $N_0$
and $2N_0$ there are at most $ \sim N_0^{2/3 + \delta}$ optimal 
rectangles   and at least  $ \sim N_0^{2/3 - \delta}$ optimal 
rectangles.  To derive a fluctuation estimate from this, we need another property, namely that for each optimal rectangle there is another optimal rectangle whose  particle  number differs in a quantified way.

\begin{lemma}\label{lemma: counting2}
For   $0<\delta<\frac{1}{6}$ there exists $ N_{\beta,\delta}  \in \N$ such that for all $ N_0  \ge N_{\beta,\delta}  $ the following holds. Let $h,l \in \mathbb{N}$ be such that $hl=N_0$,  $c_\beta\sqrt{N_0} \le h \le C_\beta\sqrt{N_0}$,      and assume that the $(h,l)$-rectangle is optimal. Then,   
\begin{itemize}
\item[(i)]  there exist $h_+,l_+\in \mathbb{N}$ such that the $(h_+,l_+)$-rectangle is optimal and for $N_+ = h_+ l_+$ we have
\begin{align*}
0< N_+-N_0\le  N_0^{1/3+\delta}\,;
\end{align*}
\item[(ii)]   there exist $h_-,l_-\in \mathbb{N}$ such that the $(h_-,l_-)$-rectangle is optimal and for $N_- = h_- l_-$ we have 
\begin{align*}
0< N_0-N_-\le  N_0^{1/3+\delta}\,.
\end{align*}
\end{itemize}
\end{lemma}
 
We postpone the proofs of \Cref{lemma: counting} and~\Cref{lemma: counting2} and first explain the fluctuation estimate. 

\begin{proof}[Proof of the \Cref{thm:fluctuation},   lower bound
  {\rm ($L_2$)} and upper bound   {\rm ($U_2$)}] 
We start by a preliminary observation. Consider an optimal
$(h_0,l)$-configuration consisting of $N$  particles  with  $N
< h_0l$. Then, also the $(h_0,l)$-configuration consisting of $N+1$
 particles  (adding one  particle  in the upmost row) is
optimal. Indeed, by the monotonicity of the energy in the  particle  number
and the fact that the two configurations have the same energy we get
$\min_h \mathcal{H}_{\beta,N+1}(h) \ge \min_h \mathcal{H}_{\beta,N}(h)
=  \mathcal{H}_{\beta,N}(h_0) = \mathcal{H}_{\beta,N+1}(h_0)$. This
shows that the $(h_0,l)$-configuration with $N+1$  particles 
is optimal. Repeating this  argument  we also find that the $(h_0,l)$-rectangle   consisting of $h_0l$  particles  is optimal. 

In particular, this argument shows that  $\min \mathcal{H}_{\beta,N} <
\min \mathcal{H}_{\beta,N+1}$ is only possible if $N$ is a  particle 
number for which an optimal rectangle exists. We now proceed with the
proof of  the lower bound  {\rm ($L_2$)} and of the upper bound
  {\rm ($U_2$)}. To this end, we let $N_0 \ge  N_{\beta,\delta} $, where $ N_{\beta,\delta} $ denotes the (maximum of the) constants in \Cref{lemma: counting}--\Cref{lemma: counting2}.

%We start by a preliminary observation. Consider an optimal $(h_0,l)$-configuration consisting of $\hat{n}$  particles  with  $\hat{n} < n := h_0l$. Then, also the $(h_0,l)$-square consisting of $n$  particles  is optimal. Indeed, by the monotonicity of the energy in the  particle  number and the fact that the two configurations have the same energy we get  $\min_h E_{n}(h) \ge \min_h E_{\hat{n}}(h) =  E_{\hat{n}}(h_0) \le E_{{n}}(h_0)$. This shows that the $(h_0,l)$-square is optimal.

\emph{Proof of  the lower bound  {\rm ($L_2$)}:} Since by
\Cref{lemma: counting}(ii) and the comment below \eqref{hrange} there
are at most $ C_\beta  N_0^{2/3 + \delta}$ many optimal rectangles
with  particle  numbers between $N_0$ and $2N_0$, we can choose $N \in
[N_0,2N_0]$ such that $N' - N \ge c_\beta  N_0^{1/3-\delta}$, where
$N' > N$ denotes the smallest number bigger than $N$ such that there
exists an optimal rectangle with $N'$  particles,   and $ c_\beta >0$ is a
constant depending only on $\beta$.  The  above  preliminary
observation yields $\min \mathcal{H}_{\beta,N'} = \min
\mathcal{H}_{\beta,N}$. Consider an optimal $(h,l)$-rectangle
consisting of $N'$  particles,   and consider the corresponding
$(h,l)$-configuration consisting of $N$  particles  by removing $N'-N$  particles 
from the upmost row. Due to $\min \mathcal{H}_{\beta,N'} = \min
\mathcal{H}_{\beta,N}$, this configuration is also optimal. As $N' - N
\ge c_\beta  N_0^{1/3-\delta}$, this configuration allows for a
fluctuation of order $\sim N_0^{1/3-\delta}$. Since this holds for all
$N_0 \ge N_{\beta,\delta}  $,  the lower bound  {\rm ($L_2$)} follows for a constant $c_{\beta,\delta}>0$ depending on $\beta$ and $\delta$. 

\emph{Proof of  the upper bound  {\rm ($U_2$)}:} Fix $N \in
[N_0,2N_0]$ and consider an optimal $(h,l)$-configuration consisting
of $N$  particles.    By the  above  preliminary observation and 
by   \Cref{lemma: counting2}   we find that the $(h,l)$-rectangle
with  particle  number $N' = hl$  is optimal and it holds that    $N' - N
\le N_0^{1/3+\delta}  $. (Indeed,  if  $N' - N >
N_0^{1/3+\delta}$, ~\Cref{lemma: counting2}(ii)  (applied for
$N'$)    would imply that there exists $N < N'' <N'$ such that
there exists an optimal rectangle with $N''$  particles.   As $\min
\mathcal{H}_{\beta,N'} = \min \mathcal{H}_{\beta,N}$, we would get $
\min \mathcal{H}_{\beta,N'}  = \min \mathcal{H}_{\beta,N''}$ by
monotonicity, which yields a contradiction as the energy for different
optimal rectangles is necessarily different.)  This shows that at most
$\sim N_0^{1/3+\delta}$  particles  are missing in the upmost row of the optimal $(h,l)$-configuration and thus  the fluctuation of the configuration is
controlled by $N_0^{1/3+\delta}$. As this holds for all $N_0 \ge N_{\beta,\delta}  $,  the upper bound  {\rm ($U_2$)} follows for a constant $C_{\beta,\delta}>0$ depending on $\beta$ and~$\delta$.
\end{proof}

We now come to the proofs of \Cref{lemma: counting} and \Cref{lemma: counting2}.
\begin{proof}[Proof of \Cref{lemma: counting}]
Let $0<\delta<\frac{1}{6}$. Along the proof we will progressively 
increase the  lower bound $ N_{\beta,\delta} $, depending on $\delta$ and $\beta$. 
 Fix a  $(h_0,l)$-rectangle, and set $N = h_0l$.   Given $h
\in \N$, we first rewrite the energy $\mathcal{H}_{\beta,N}$ defined
in \eqref{def:HL}   by   easily computing    
\begin{align*}
 \mathcal{H}_{\beta,N}(h) & = 2 h  + 2(1-\beta)   \Big\lceil \frac{h_0 l }{h} \Big\rceil =  2h_0 + 2(1-\beta)  l + 2(h-h_0) +  2(1-\beta)   \Big\lceil \frac{h_0 l }{h} - l \Big\rceil  \\
 & =   2h_0 +  2(1-\beta)  l + 2(h-h_0) + 2(1-\beta)   \Big\lceil \frac{l }{h_0}(h_0-h) + \frac{l}{h h_0} (h_0-h)^2 \Big\rceil\\ 
  & =   \mathcal{H}_{\beta,N}(h_0)   + 2(h-h_0) +  2(1-\beta)   \big\lceil   t(h_0,h,l)  \big\rceil\,. 
\end{align*}
 where for shorthand we have  set 
\begin{align}\label{t-def}
t \defas t(h_0,h,l) \defas   \frac{1}{1-\beta }    (h_0-h)  + \Big(   \frac{l }{h_0} -  \frac{1}{1-\beta }     \Big) (h_0-h)   + \frac{l}{h h_0} (h_0-h)^2  . 
\end{align}
We observe that 
\begin{align}\label{show1}
\lceil t \rceil \ge   \Big\lceil  \frac{1}{1-\beta }    (h_0-h)\Big\rceil \text{ for all $h$}   \quad \Rightarrow  \quad \text{$h_0 = {\rm argmin}\, \mathcal{H}_{\beta,N} (h)$ },
\end{align}
\begin{align}\label{show2}
\lceil t \rceil \le  \Big \lceil  \frac{1}{1-\beta }   (h_0-h)\Big\rceil-1  \text{ for some $h$}   \quad \Rightarrow  \quad \text{ $h_0 \neq {\rm argmin}\,  \mathcal{H}_{\beta,N} (h)$.}
\end{align}

\emph{Proof of {\rm (i)}:}  Consider  $l \in \N$ with  $|l-l(h_0)| \le N_0^{1/6 - \delta}$.  In view of \eqref{show1}, for fixed $h$,  it suffices to show that $\lceil t \rceil \ge   \lceil  \frac{1}{1-\beta }    (h_0-h)\rceil$.   Since $\beta$ is algebraic, by Roth's theorem, see~\cite{Roth}, we find $c=c(\beta,\delta)>0$ such that 
$${\Big| \frac{1}{1-\beta}     - \frac{p}{q}\Big| \ge c  q^{-2-\delta} }$$
for all  $p \in \Z$ and $q \in \N$. For the choice $q = |h_0-h|$  this yields
$ |  \frac{1}{1-\beta }    |h_0-h|  - p| \ge c |h-h_0|^{-1-\delta} $ for all $p \in \Z$. Thus, writing $x = h_0 - h$, we get from \eqref{t-def}
$$ t \ge   \Big\lfloor  \frac{1}{1-\beta }    x   \Big\rfloor  + c |x|^{-1-\delta} + \Big(   \frac{l }{h_0} -  \frac{1}{1-\beta }     \Big)x   + \frac{l}{h h_0} x^2.  $$ 
By using \eqref{best approx} and the fact that $|l-l(h_0)| \le N_0^{1/6-\delta}$ we further get
$$ t \ge   \Big\lfloor  \frac{1}{1-\beta }    (h-h_0)   \Big\rfloor  + c |x|^{-1-\delta} - \Big(1 + N_0^{1/6-\delta}) \frac{1}{h_0} |x|   + \frac{l}{h h_0} x^2.  $$ 
It is not restrictive to assume that $h \in  [  \frac{1}{2}  c_\beta  N_0^{1/2},  2  C_\beta  N_0^{1/2}]$ as otherwise $h$ cannot be a minimizer, see \eqref{hrange}. We have $h_0 \in [ \frac{1}{2}  c_\beta N_0^{1/2},  2  C_\beta  N_0^{1/2}]$ by assumption and therefore also $l \ge   \frac{1}{2}  c_\beta  N_0^{1/2}$, cf.\ \eqref{best approx}. This shows 
$${c |x|^{-1-\delta} - \Big(1 + N_0^{1/6-\delta}\Big) \frac{1}{h_0} |x|   + \frac{l}{h h_0} x^2   \ge c |x|^{-1-\delta} -   4  c_\beta^{-1}  N_0^{-1/3-\delta} |x|   +  \frac{1}{8}  c_\beta C_\beta^{-2}  N_0^{-1/2} x^2.  } $$
If $|x| >  32   (C_\beta c_\beta^{-1})^2  N_0^{1/6-\delta}$,
the  sum of the  last two terms on the right-hand side  is
 positive. If $|x| < ( \frac{c_\beta c}{ 4  }
N_0^{1/3+\delta})^{1/(2+\delta)}$, the  sum of the  first two
terms on the right-hand side  is   positive. Since   $32  (C_\beta c_\beta^{-1})^2  N_0^{1/6-\delta} <   (  \frac{c_\beta c}{4}   N_0^{1/3+\delta})^{1/(2+\delta)} $ for $N_0 \ge N_{\beta,\delta}  $ for some $ N_{\beta,\delta}  $ depending on $\delta$ and $\beta$, we conclude
$$ t >   \Big\lfloor  \frac{1}{1-\beta }    (h-h_0)   \Big\rfloor.$$
This shows $ \lceil t \rceil \ge   \lceil   \frac{1}{1-\beta }    (h-h_0) \rceil $ and concludes the proof of (i).

\emph{Proof of {\rm (ii)}:}  Consider   $l \in \N$ with  $|l-l(h_0)| \ge N_0^{1/6 + \delta}$. Without restriction we treat the case $l < l(h_0)$ as the other one is completely analogous.  We indicate the adaptations at the end of the proof.   In view of \eqref{show2}, it suffices to find $h$ such that the corresponding $t$ defined in \eqref{t-def} satisfies  $\lceil t \rceil \le   \lceil  \frac{1}{1-\beta }    (h_0-h)\rceil - 1$. Choose $\alpha $ with  $ \frac{1/3 - \delta}{2-4\delta} < \alpha < \frac{1/6 + \delta}{1+3\delta} $.     By~\Cref{lemma:discr}(i) below, we can find   $x \in  \N $ with     $N_0^{\alpha(1-3\delta)} \le x \le N_0^\alpha$ such that  $  \frac{1}{1-\beta }    x - \lfloor   \frac{1}{1-\beta }    x \rfloor  \le N_0^{-\alpha (1-\delta)}$. We  set $h  =  h_0-x$ and aim to show that the corresponding $t$ satisfies $\lceil t \rceil \le   \lceil  \frac{1}{1-\beta }    (h_0-h)\rceil - 1$. Note that we can assume without restriction that $l \le  C_\beta  N_0^{1/2}$ for some $C_\beta>0$ as otherwise the configuration is clearly not optimal. Using also  $h,h_0 \ge   \frac{1}{2}  c_\beta  N_0^{1/2}$, cf.\ below \eqref{hrange}, we estimate
$$ { t \le  \Big\lfloor  \frac{1}{1-\beta }    (h_0-h) \Big\rfloor + N_0^{-\alpha (1-\delta)} + \Big(   \frac{l }{h_0} -  \frac{1}{1-\beta }     \Big) x  +   4  C_\beta c_\beta^{-2}  N_0^{-1/2} x^2 . } $$ 
Using $x>0$,  $\frac{l }{h_0} -  \frac{1}{1-\beta }     \le (- N_0^{1/6+\delta}+1) h_0^{-1} $ (see \eqref{best approx} and use the choice of $l$), as well as $h_0 \le  2  C_\beta  N_0^{1/2}$ we get
$$ t \le  \Big\lfloor  \frac{1}{1-\beta }    (h_0-h) \Big\rfloor + N_0^{-\alpha (1-\delta)}  - \frac{1}{ 4  } C_\beta^{-1}  N_0^{-1/3+\delta} x   +  4   C_\beta c_\beta^{-2}  N_0^{-1/2} x^2 ,  $$
provided that $ N_{\beta,\delta}  $ is chosen sufficiently large. Recalling that   $N_0^{\alpha(1-3\delta)} \le x \le N_0^\alpha$ and $ \frac{1/3 - \delta}{2-4\delta} < \alpha < \frac{1/6 + \delta}{1+3\delta} $,  for $N_0 \ge  N_{\beta,\delta}  $ and $ N_{\beta,\delta}  $ sufficiently large,  we find by an elementary computation that 
 $$t \le \Big\lfloor  \frac{1}{1-\beta }    (h_0-h) \Big\rfloor. $$
This  yields   $\lceil t \rceil \le   \lceil  \frac{1}{1-\beta }    (h_0-h)\rceil - 1$ and concludes the proof. 

 In the case $l >  l(h_0)$, we use~\Cref{lemma:discr}(ii) below to find $x \in \N$ with  $N_0^{\alpha(1-3\delta)} \le x \le N_0^\alpha$ and  $  \frac{1}{1-\beta }     (- x) - \lfloor   \frac{1}{1-\beta }     (-x) \rfloor  \le N_0^{-\alpha (1-\delta)}$, and define $h = h_0 + x$. Then, the statement follows by repeating the estimates above.  
\end{proof}

We check the following result which has been used in the previous proof and will be also instrumental in the proof of  \Cref{lemma: counting2} below.  

\begin{lemma}\label{lemma:discr}
Let $0<\delta<\frac{1}{6}$ and $0<\alpha < 1$. There exist $ N_{\beta,\delta}  \in \N$  (depending on $\delta$,  $\alpha$, and $\beta$) and $c>0$ (depending on $\delta$ and $\beta$) such that for $N_0 \ge N_{\beta,\delta}  $ 
\begin{itemize}
\item[(i)]  we find   $x \in  \N $ with  $N_0^{\alpha(1-3\delta)} \le x \le N_0^\alpha$ such that 
\begin{align}\label{XXX}
 c N_0^{-\alpha(1+\delta)} \leq  \frac{1}{1-\beta }    x - \left\lfloor  \frac{1}{1-\beta }    x \right\rfloor  \le N_0^{-\alpha (1-\delta)}\,;
\end{align}
\item[(ii)]  we find   $x \in  \N $ with  $N_0^{\alpha(1-3\delta)} \le x \le N_0^\alpha$ such that 
\begin{align*}
 c N_0^{-\alpha(1+\delta)} \leq \left\lceil  \frac{1}{1-\beta }    x \right\rceil -  \frac{1}{1-\beta }   x  \le N_0^{-\alpha (1-\delta)}\,.
\end{align*}

\end{itemize}

\end{lemma}

\begin{proof} First we show {\rm (i)}.
We define the $m$-discrepancy of the sequence $(j  \frac{1}{1-\beta}     \, {\rm mod} \, 1)_{j \in \N}$ with respect to the interval $[0,s]$, $s>0$, by
$${\phi_m(s) := \frac{1}{m} \sum_{j=0}^{m-1}  \chi_{\N + [0,s]}  (j \tfrac{1}{1-\beta})   -s\,.} $$
As $\beta$ is algebraic, by \cite[Thm.\ 3.2 and Ex.\ 3.1, pp.\ 123-124]{Kuipers},  for $\delta >0$ there exists $C_{\delta,\beta}>0$ depending on $\delta$ and $\beta$ such that
$$\sup_s |\phi_m(s)| \le C_{\delta,\beta} m^{-1+\delta/2} \quad \text{for all $m \in \N$}. $$
With $s =  N_0^{-\alpha(1-\delta)}$ and $m = \lfloor N_0^{\alpha} \rfloor$ we get 
$$ {\frac{1}{m} \sum_{j=0}^{m-1}  \chi_{\N + [0,s]}  (j \tfrac{1}{1-\beta})  \ge  N_0^{-\alpha (1-\delta)} - C_{\delta,\beta} m^{-1+\delta/2} >0\,, }$$
provided that $N_0 \ge N_{\beta,\delta} $ for some $ N_{\beta,\delta}  $ depending on $\delta$, $\alpha$, and $\beta$. 
This shows that there exists $x \in \N_0$ such that $x \le N_0^\alpha$ and  $  \frac{1}{1-\beta }    x - \lfloor   \frac{1}{1-\beta }   x \rfloor  \le s= N_0^{-\alpha (1-\delta)}$,  i.e., the upper bound in \eqref{XXX} holds. By Roth's Theorem~\cite{Roth} there exists $c>0$ depending on $\delta$ and $\beta$ such that $|  q \frac{1}{1-\beta}      - p  | > \frac{c}{q^{1+\delta}}$ for all $p \in \Z$ and $q \in \N$. Thus,  $| \frac{1}{1-\beta }    x - \lfloor   \frac{1}{1-\beta }    x \rfloor | > \frac{c}{x^{1+\delta}} \geq cN_0^{-\alpha(1+\delta)}$  which yields  the lower bound    in \eqref{XXX}.  Furthermore, as this also implies $\frac{x^{1+\delta}}{c} \ge N_0^{\alpha (1-\delta)}$  by the upper bound in \eqref{XXX}, we obtain $x \ge N_0^{\alpha (1-3\delta)}$ for $N_0 \ge  N_{\beta,\delta} $, provided that $ N_{\beta,\delta} $ is chosen sufficiently large.

The proof of {\rm (ii)} follows along the same lines by using the $m$-discrepancy of the sequence $(j\frac{1}{1-\beta}     \, {\rm mod} \, 1)_{j \in \N}$ with respect to the interval $[1-s,1]$, $s>0$, namely
$$\hat\phi_m(s) := \frac{1}{m} \sum_{j=0}^{m-1}  \chi_{\N + [1-s,1]}  (j \tfrac{1}{1-\beta})   -s\,. $$
 Using again \cite{Kuipers} we can control $\hat\phi_m$ in the same way as $\phi_m$ above.  
\end{proof}

 We prove now \Cref{lemma: counting2}.

\begin{proof}[Proof of \Cref{lemma: counting2}] Let $h,l \in
  \mathbb{N}$, $N_0  =hl$, and suppose that the  $(h,l)$-rectangle is
  optimal. We only prove ({\rm i}) as the proof of ({\rm ii}) is
  completely analogous  and we  briefly indicate the necessary
   adaptations  at the end of the proof.   Let $\sigma_l = 1$ if $l- l(h)\ge 0$ and $\sigma_l = -1$ else.    We apply~\Cref{lemma:discr} ((i) or (ii), depending on the sign of $\sigma_l$)  for $\alpha = \frac{1}{6}-3\delta$ to choose $p,q \in \N$ with
  \begin{align}\label{ineq:prop-qp}
  N_0^{(\frac{1}{6}-3\delta)(1-3\delta)} \leq q \leq N_0^{\frac{1}{6}-3\delta} \quad \text{and}\quad c
  q\leq p \leq C q 
  \end{align}
  such that
  \begin{align}\label{ineq:approx}
c N_0^{-(\frac{1}{6}-3\delta)(1+\delta)} \leq  \sigma_l  \Big( q \frac{1}{1-\beta}      - p \Big) \leq N_0^{-(\frac{1}{6}-3\delta)(1-\delta)}\,.
  \end{align} 
   We claim that
  \begin{align}\label{impi}
cN_0^{-\frac{1}{3}+ 5  \delta}\leq    \sigma_l   \Big( \frac{l}{h} - \frac{p}{q} \Big) \leq CN_0^{-\frac{1}{3} + 7  \delta}\,.
  \end{align}
Indeed, by \eqref{best approx},  \eqref{ineq:prop-qp}, \eqref{ineq:approx},  and the fact that $|l - l(h)| \le N_0^{1/6+\delta}$ (see \Cref{lemma: counting}(ii)), we obtain
\begin{align*}
  \sigma_l   \Big( \frac{l}{h} - \frac{p}{q} \Big) & \le \frac{1+N_0^{1/6+\delta}}{h} +  \sigma_l   \Big( \frac{1}{1-\beta}     - \frac{p}{q}  \Big) \le \frac{CN_0^{1/6+\delta}}{h} +   N_0^{-(\frac{1}{6}-3\delta)(1-\delta)} q^{-1} \\
&\le C   N_0^{-1/3+\delta} + C N_0^{-1/3 + \frac{20}{3}\delta-12\delta^2  } \le CN_0^{-1/3 + 7  \delta}\,,
\end{align*}
where the constant depends on $\delta$ and $\beta$. Here, we used that $h \ge c_\beta N_0^{1/2}$.   On the other hand,   using again $h \ge c_\beta N_0^{1/2}$,
\begin{align*}
 \sigma_l   \Big( \frac{l}{h} - \frac{p}{q} \Big) & \ge -\frac{1+N_0^{1/6+\delta}}{h} +  \sigma_l   \Big( \frac{1}{1-\beta}      - \frac{p}{q} \Big) \ge -\frac{CN_0^{1/6+\delta}}{h} +   cN_0^{-(\frac{1}{6}- 3 \delta)(1+\delta)} q^{-1} \\
&\ge -CN_0^{-1/3+\delta} + C N_0^{-1/3 + \frac{35}{6}\delta+3\delta^2  } \ge CN_0^{-1/3 + 5  \delta}\,.
\end{align*}
 This shows \eqref{impi}.  Recalling   that   $|l-l(h)| \le N_0^{1/6+\delta}$ and  $ p \le Cq \le C N_0^{1/6-3\delta}$, we can further  choose $k \in \N$, $k \le C N_0^{1/6+\delta}/q $, such that 
\begin{align}\label{kkk}
\big| 2(1-\beta) |l-l(h)|  - 2kq - 2(1-\beta) kp \big| \le CN_0^{1/6-3\delta}
\end{align}
for  some   universal $C>0$. 
We now define $l_+ = l  - \sigma_l  kp$, $h_+ = h +  \sigma_l  kq$,
and $N_+=l_+h_+$.  We  claim  that 
\begin{align}\label{ineq:cardestimateNplus}
0 < N_+-N_0 \leq C  N_0^{\frac{1}{3} + 8  \delta}\,, 
\end{align}
 provided that $ N_{\beta,\delta} $ is sufficiently large.  Indeed, 
\begin{align*}
N_+-N_0= h_+l_+ -hl = (h +  \sigma_l  kq)(l -  \sigma_l  kp) =     \sigma_l  \big( k ql-k  h  p \big)   -k^2pq .
\end{align*}
By \eqref{impi} we find
\begin{align*}
kq \big(c  h  N_0^{-\frac{1}{3}+ 5  \delta} - kp\big)\leq   \sigma_l   \big( k ql-k h   p \big)   -k^2pq \leq CN_0^{-\frac{1}{3} +  7  \delta} k q h -k^2 pq\,.
\end{align*}
Using that $ c_\beta  N_0^{\frac{1}{2}}\le h \le  C_\beta  N_0^{1/2}$,
and  that  $kp, kq \le CN_0^{1/6+\delta}$ we get 
\begin{align*}
0<kq  (cc_\beta N_0^{\frac{1}{6}+ 5 \delta} -CN_0^{\frac{1}{6}+ \delta  }   )\leq N_+ - N_0 \le  CC_\beta    N_0^{\frac{1}{3} + 8  \delta}\,,
\end{align*}
where the first  inequality holds for $N_0\ge N_{\beta,\delta} $ for some $ N_{\beta,\delta}  $ sufficiently large depending on $\delta$  and $\beta$.  This concludes the proof of \eqref{ineq:cardestimateNplus}.  We observe that the desired estimate follows from \eqref{ineq:cardestimateNplus} by replacing $\delta$ with $ \delta/8  $. 

It remains to check that also the   $(h_+, l_+)$-rectangle is optimal. First note that  $ c_\beta N_0^{1/2} \le h \le  C_\beta  N_0^{1/2} $ implies $\frac{c_\beta}{2} N_0^{1/2} \le h_+ \le 2C_\beta N_0^{1/2} $ for $N_0 \ge N_{\beta,\delta} $ and $ N_{\beta,\delta} $ sufficiently large as $kq \le  C N_0^{1/6+\delta}$.   Thus,  in  view of ~\Cref{lemma: counting}(i),  it suffices to check that  
\begin{align}\label{suffice}
|2 h_+ - 2(1-\beta)  l_+| \le  \frac{1}{2}N_0^{1/6- 2    \delta}, 
\end{align}
as then, by using \eqref{best approx}, we get
$$    2(1-\beta)    |l_+ - l(h_+)| \le | 2(1-\beta)  l_+ - 2h_+| +   |
2(1-\beta)  l(h_+) - 2h_+| \le  \frac{1}{2}N_0^{1/6-  2   \delta}
+2\le   2(1-\beta)   N_{ \delta}^{1/6-\delta}\,,    $$
where the last inequality holds for $ N_{\beta,\delta} $ sufficiently large depending on $\delta$ and $\beta$. 
Let us finally  show \eqref{suffice}. By \eqref{best approx} and \eqref{kkk}, we calculate
\begin{align*}
\big|2 h_+ - 2(1-\beta)  l_+\big| & =   \big|2 h - 2(1-\beta)  l(h)    + 2(1-\beta)  (l(h) - l) + 2 (h_+ -  h)   -  2(1-\beta)  (l_+-l) \big| \\
& \le 2 +  \big| 2(1-\beta) (l(h) - l) + 2kq  \sigma_l   + 2(1-\beta) kp  \sigma_l   \big|  \\ 
& \le 2 + CN_0^{1/6- 3  \delta} \le \frac{1}{2}N_0^{1/6 -  2    \delta}  
\end{align*}
for $N_0\ge N_{\beta,\delta} $ and $ N_{\beta,\delta} $ sufficiently large depending on $\delta$. Here, we used that $\sigma_l = 1$ if $l- l(h)\ge 0$ and $\sigma_l = -1$ else.   This shows \eqref{suffice} and concludes the proof.

 Finally, the proof of (ii) follows along the same lines by choosing $p$ and $q$ such that \eqref{ineq:approx} is replaced by  $c N_0^{-(\frac{1}{6}-3\delta)(1+\delta)} \leq   \sigma_l  ( -q \frac{1}{1-\beta}      + p ) \leq N_0^{-(\frac{1}{6}-3\delta)(1-\delta)}$.   
\end{proof}

\section*{Acknowledgements} 
This work was supported by the DFG project FR 4083/3-1 and by the
Deutsche Forschungsgemeinschaft (DFG, German Research Foundation)
under Germany's Excellence Strategy EXC 2044 -390685587, Mathematics
M\"unster: Dynamics--Geometry--Structure. L.K. was supported by the
DFG through the Emmy Noether Programme (project number
509436910). U.S.  was partially funded by the Austrian Science Fund (FWF) projects 10.55776/F65,  10.55776/I5149. For
open-access purposes, the authors have applied a CC BY public copyright
license to any author-accepted manuscript version arising from this
submission. 

\bigskip

\section*{Conflict of interest}

The authors have no competing interests to declare that are relevant to the content of this article.

 Data sharing not applicable to this article as no datasets were generated or analysed during the current study.

 \appendix 
\section{Proof of \Cref{lemma:elementaryprop}}\label{appendix}

\begin{proof}[Proof of \Cref{lemma:elementaryprop}] Let $C_n$ be a minimizer of \eqref{eq: main energy}. For simplicity, we write $G=(V,E)$ instead of $G_{\mathrm{nat}}=(V,E_{\mathrm{nat}})$ for the associated natural bond graph.   \\
\noindent \begin{step}{1}
 We define $V_{\rm sub} = V \cup \bigcup_{x\in V}  \mathcal{N}_{\mathcal{L}^-}(x)$ and $E_{\rm sub} = E \cup \lbrace \lbrace x,y \rbrace \colon x \in V, y\in \mathcal{N}_{\mathcal{L}^-}(x) \rbrace$.  In this step, we show
\begin{align}\label{ineq:1epsregular}
|x-y| \ge  1-\varepsilon \text{ for all }  \lbrace x,y \rbrace \in E_{\rm sub}\,. 
\end{align}
 We define  
\begin{align}\label{def:M}
M:= \max_{x \in \mathbb{R}^2} \#(V_{\rm sub}   \cap B_{\frac{1}{2}(1-\varepsilon)}(x))\,,
\end{align}
 where $B_r(x)$ denotes the closed ball with radius $r>0$ centered at $x$ and we write $B_r$ in the following whenever $x=0$.     Claim \eqref{ineq:1epsregular} will follow by showing  $M=1$. Let $ x_0 \in \R^2$ be a maximizer in \eqref{def:M}. After translation of $V_{\rm sub}$, it is not restrictive to  assume  that $x_0=0$.   As $\#\{\{x,y\} \in E_{\rm sub} \colon x,y \in B_{\frac{1}{2}(1-\varepsilon)}\} \geq \frac{M(M-1)}{2}$, by assumption {\rm ($\rm{iii}_2$)}, letting $\lambda = \min\{1, 2  \beta\}$,  we have 
\begin{align}\label{estimate:insideBall}
\underset{\{x,y\} \in E}{\sum_{x,y \in B_{\frac{1}{2}(1-\varepsilon)}}} v_2(|x-y|) +  2  \beta  \underset{x \in V, y\in \mathcal{N}_{\mathcal{L}^-}(x) }{\sum_{x,y\in  B_{\frac{1}{2}(1-\varepsilon)} }}  v_2(|x-y|)   \geq \frac{\lambda}{2\varepsilon}M(M-1)\,.
\end{align}
Consider the annulus $ A_\eps :=B_{\frac{1}{2}(1-\varepsilon) +\sqrt{2}} \setminus B_{\frac{1}{2}(1-\varepsilon)} \subset B_{\frac{1}{2}+\sqrt{2}}$. By covering $B_{\frac{1}{2}+\sqrt{2}}$ with discs of radius $\frac{1}{4}$ we get that there exists $ K \in  \mathbb{N}$ and $\{z_i\}_{i=1}^{K} \subset \mathbb{R}^2$ such that  for all $0 < \eps \le \frac{1}{2}$ 
\begin{align*}
A_\eps  \subset B_{\frac{1}{2}+\sqrt{2}}  \subset  \bigcup_{i=1}^{K} B_{\frac{1}{4}}(z_i)   \subset \bigcup_{i=1}^{K} B_{\frac{1}{2}(1-\varepsilon)}(z_i)\,.
\end{align*}
Note that $K$ is independent of $\varepsilon$ (provided $\varepsilon <\frac{1}{2}$). Recalling \eqref{def:M}, we have
\begin{align}\label{ineq:cardannulus}
\#(V_{\rm sub} \cap A_\eps) \leq \#\left(V_{\rm sub}\cap  \bigcup_{i=1}^K B_{\frac{1}{2}(1-\varepsilon)}(z_i)\right)\leq \sum_{i=1}^K \#\left(V_{\rm sub}  \cap  B_{\frac{1}{2}(1-\varepsilon)}(z_i)\right)\leq K M\,.
\end{align}
By \eqref{ineq:cardannulus}, the definition of $M$, and {($\rm i_2$)}, letting $\Lambda = \max\{1,  \beta\}$ we have
\begin{align}\label{ineq:annulusinteraction}
-\underset{\{x,y\} \in E}{\sum_{x \in B_{\frac{1}{2}(1-\varepsilon)}, y \in A_\eps}} v_2(|x-y|)  -   \beta  \underset{x \in V, y\in \mathcal{N}_{\mathcal{L}^-}(x) }{\sum_{x \in B_{\frac{1}{2}(1-\varepsilon)}, y \in A_\eps}}  v_2(|x-y|)  & \leq \Lambda \cdot \#\{x \in V_{\rm sub} \cap B_{\frac{1}{2}(1-\varepsilon)}, y \in V_{\rm sub} \cap A_\eps\} \notag \\ & \leq  \Lambda KM^2\,.
\end{align}
We write $V\cap B_{\frac{1}{2}(1-\varepsilon)} =  \{x_i\}_{i=1}^{\bar{M}} $,  $\bar{M}\le M$, and consider  a competitor $\hat{V}$ (with associated natural  bond graph $\hat{G}$) given by
\begin{align*}
\hat{V} = (V \setminus B_{\frac{1}{2}(1-\varepsilon)}) \cup \bigcup_{i=1}^M \{x_i +\tau_i\}\,,
\end{align*}
where $\tau_i \in \mathbb{R}^2$ are chosen such that
\begin{align}\label{ineq:choicetaui}
\mathrm{dist}(x_i+\tau_i, \hat{V}\setminus \{x_i+\tau_i\}) \geq \sqrt{2} \text{ for all } i=1,\ldots, \bar{M}\,.
\end{align}
By \eqref{eq: relation}, \eqref{ineq:choicetaui}, {($\rm ii_2$)},  and the optimality of $G$ we have
\begin{align}\label{ineq:Xmin}
F_\beta  (G) \leq  F_\beta  (\hat{G}) &\leq F_\beta  (G) - \underset{\{x,y\} \in E}{\sum_{x,y \in B_{\frac{1}{2}(1-\varepsilon)}}} v_2(|x-y|)-   2  \beta  \underset{x \in V, y\in \mathcal{N}_{\mathcal{L}^-}(x) }{\sum_{x,y\in  B_{\frac{1}{2}(1-\varepsilon)} }}  v_2(|x-y|)  \notag \\
& \quad  - 2 \underset{\{x,y\} \in E}{\sum_{x \in B_{\frac{1}{2}(1-\varepsilon)}, y \in A_\eps}} v_2(|x-y|) -  2  \beta  \underset{x \in V, y\in \mathcal{N}_{\mathcal{L}^-}(x) }{\sum_{x \in B_{\frac{1}{2}(1-\varepsilon)}, y \in A_\eps}}  v_2(|x-y|)  \,.
\end{align}
 Now, using  \eqref{estimate:insideBall}, \eqref{ineq:annulusinteraction}, and \eqref{ineq:Xmin}, we obtain
\begin{align*}
\frac{\lambda  }{2\varepsilon}M(M-1)  \leq 2 \Lambda K M^2\,.
\end{align*}
For $\varepsilon>0$ small enough ($\varepsilon < \frac{\lambda}{5\Lambda K  }  $ suffices), this inequality can only be true for $M=1$. This yields \eqref{ineq:1epsregular} and concludes Step 1.
\end{step} \\
\noindent \begin{step}{2} In this step we prove that all bond angles satisfy 
\begin{align}\label{incl:2epsregular}
\theta \in [\pi/2-\varepsilon,\pi/2+\varepsilon] \cup [\pi-\varepsilon,\pi+\varepsilon] \cup [3\pi/2-\varepsilon,3\pi/2+\varepsilon]\,.
\end{align}
In particular, by choosing $\varepsilon< \frac{1}{10}\pi$, this will also imply that $\#\mathcal{N}(x,E) \leq 4$ and $\# \mathcal{N}_{\mathcal{L}^-}(x) \le 1$ for all $x \in V$,  where we recall that in \eqref{eq: main energy} the energy contributions is due to triples $(x,y,z) \in  \left(C_N\cup\mathcal{L}^-\right)^3$. To see \eqref{incl:2epsregular}, we   first    claim that 
\begin{align}\label{ineq:neighbourhood}
\#\mathcal{N}(x,E) + \# \mathcal{N}_{\mathcal{L}^-}(x) \leq 4\frac{\left(\sqrt{2} +\frac{1}{2}\right)^2}{(1-\varepsilon)^2} \text{ for all } x\in V\,.
\end{align} 
Indeed,  by Step~1,  {\rm ($\rm{ii}_2$)}, and by the fact that $B_{\frac{1}{2}(1-\varepsilon)}(y) \subset B_{\sqrt{2}+\frac{1}{2}}(x)$ for all $y \in \mathcal{N}(x,E) \cup \mathcal{N}_{\mathcal{L}^-}(x) $ we have
\begin{align*}
\big(\sqrt{2} +\tfrac{1}{2}\big)^2\pi = \big| B_{\sqrt{2}+\frac{1}{2}}(x)\big| \geq \sum_{y \in \mathcal{N}(x,E) \cup \mathcal{N}_{\mathcal{L}^-}(x)  } \big|B_{\frac{1}{2}(1-\varepsilon)}(y)\big| \geq  \frac{1}{4}(1-\varepsilon)^2 \pi\, \big(\# \mathcal{N}(x,E) + \# \mathcal{N}_{\mathcal{L}^-}(x)  \big) \,,
\end{align*} 
i.e., \eqref{ineq:neighbourhood} holds. We now show  \eqref{incl:2epsregular}. In fact,  assuming by contradiction that  $x$ has a bond angle $\theta_{y,x,z}$ that does not satisfy \eqref{incl:2epsregular}, we could define $\hat{V}= (V\setminus \{x\} )\cup \{x+\tau\}$ for some $\tau \in \mathbb{R}^2$ such that $\mathrm{dist}(x+\tau, \hat{V}\setminus \{x+\tau\})\geq \sqrt{2}$. Then, by  {($\rm i_2$)}, {($\rm iv_3$)}, and \eqref{ineq:neighbourhood}  we obtain a contradiction to the minimality of $G$, namely
\begin{align*}
F_\beta(\hat{G}) \leq F_\beta(G) + \max\{1,  2\beta\}  \big( \#\mathcal{N}(x,E)  + \# \mathcal{N}_{\mathcal{L}^-}(x)  \big)- v_3(\theta_{y,x,z}) <F_\beta(G)\,.
\end{align*}
 Summarizing, with choosing $\varepsilon_0 : = \min\{\frac{1}{10}\pi,  \frac{\lambda}{8K\Lambda}  \}$,  the statement holds. 
\end{step}
 \end{proof}

%%%%%%%%%%%%%%%%%%%%%%%%%%%%%%%%%%%%%%%%%%%%%%

\end{document}